%% file: spca-main.tex
\documentclass{article}
\input{preamble}
\usepackage{fullpage}
\usepackage{natbib}

\title{Approximation Algorithms for Sparse Principal Component Analysis}
\author{Agniva Chowdhury\thanks{Department of Statistics, Purdue University, West Lafayette, IN, USA, 47906. \texttt{chowdhu5@purdue.edu}}
\and
Petros Drineas\thanks{Department of Computer Science, Purdue University, West Lafayette, IN, USA, 47906. \texttt{pdrineas@purdue.edu}}
\and
David P. Woodruff\thanks{School of Computer Science, Carnegie Mellon University, Pittsburgh, PA, USA, 15213. \texttt{dwoodruf@andrew.cmu.edu}}
\and
Samson Zhou\thanks{School of Computer Science, Carnegie Mellon University, Pittsburgh, PA, USA, 15213. \texttt{samsonzhou@gmail.com}}
}

\begin{document}

\maketitle

\begin{abstract}
Principal component analysis (PCA) is a widely used dimension reduction technique in machine learning and multivariate statistics. To improve the interpretability of PCA, various approaches to obtain sparse principal direction loadings have been proposed, which are termed Sparse Principal Component Analysis (SPCA). In this paper, we present thresholding as a provably accurate, polynomial time, approximation algorithm for the SPCA problem, without imposing any restrictive assumptions on the input covariance matrix. Our first thresholding algorithm using the Singular Value Decomposition is conceptually simple; is faster than current state-of-the-art; and performs well in practice. On the negative side, our (novel) theoretical bounds do not accurately predict the strong practical performance of this approach.  The second algorithm solves a well-known semidefinite programming relaxation and then uses a novel, two step, deterministic thresholding scheme to compute a sparse principal vector. It works very well in practice and, remarkably, this solid practical performance is accurately predicted by our theoretical bounds, which bridge the theory-practice gap better than current state-of-the-art. 
\end{abstract}

\section{Introduction}
Principal Component Analysis (PCA) and the related Singular Value Decomposition (SVD) are fundamental data analysis and dimensionality reduction tools in a wide range of areas including machine learning, multivariate statistics and many others. 
These tools return a set of orthogonal vectors of decreasing importance that are often interpreted as fundamental latent factors that underlie the observed data. 
Even though the vectors returned by PCA and SVD have strong optimality properties, they are notoriously difficult to interpret in terms of the underlying processes generating the data~\citep{MD09}, since they are linear combinations of \textit{all} available data points or \textit{all} available features. 
The concept of Sparse Principal Components Analysis (SPCA) was introduced in the seminal work of~\citep{AGJ2007}, where sparsity constraints were enforced on the singular vectors in order to improve interpretability. 
A prominent example where sparsity improves interpretability is document analysis, where sparse principal components can be mapped to specific topics by inspecting the (few) keywords in their support~\citep{AGJ2007,MD09,PDK2013}.

Formally, given a positive semidefinite (PSD) matrix $\Ab \in \mathbb{R}^{n \times n}$, SPCA can be defined as follows:\footnote{Recall that the $p$-th power of the $\ell_p$ norm of a vector $\x \in \mathbb{R}^n$ is defined as $\|\xb\|_p^p = \sum_{i=1}^n \abs{\x_i}^p$ for $0<p<\infty$. For $p=0$, $\|x\|_0$ is a semi-norm denoting the number of non-zero entries of $x$.}
\begin{flalign}\label{eqn:spca}
\mathcal{Z}^*= \max\limits_{\xb \in \mathbb{R}^n,\ \|\xb\|_2\leq 1}\x^\top\A\x,~~~~~\text{subject to}\,  \|\xb \|_0\le k.
\end{flalign}
In the above formulation, $\Ab$ is a covariance matrix representing, for example, all pairwise feature or object similarities for an underlying data matrix. 
Therefore, SPCA can be applied to either the object or feature space of the data matrix, while the parameter $k$ controls the sparsity of the resulting vector and is part of the input. 
Let $\xb^*$ denote a vector that achieves the optimal value ${\mathcal Z}^*$ in the above formulation. 
Intuitively, the optimization problem of eqn.~\eqref{eqn:spca} seeks a \textit{sparse}, unit norm vector $\xb^*$ that maximizes the data variance.

It is well-known that solving the above optimization problem is \textsf{NP}-hard~\citep{MWA2006} and that its hardness is due to the sparsity constraint. 
Indeed, if the sparsity constraint were removed, then the resulting optimization problem can be easily solved by computing the top left or right singular vector of $\Ab$ and its maximal value $\mathcal{Z}^*$ is equal to the top singular value of $\Ab$.

\noindent \textbf{Notation.} We use bold letters to denote matrices and vectors. 
For a matrix $\Ab \in \mathbb{R}^{n \times n}$, we denote its $(i,j)$-th entry by $A_{i,j}$; its $i$-th row by $\Ab_{i*}$, and its $j$-th column by $\Ab_{*j}$; its 2-norm by $\|\Ab\|_2 = \max_{\xb \in \mathbb{R}^n,\ \|\xb\|_2=1} \|\Ab \xb\|_2$; and its (squared) Frobenius norm by $\norm{\A}_F^2=\sum_{i,j} A_{i,j}^2$. 
We use the notation $\Ab \succeq 0$ to denote that the matrix $\Ab$ is symmetric positive semidefinite (PSD) and $\trace(\Ab) = \sum_i A_{i,i}$ to denote its trace, which is also equal to the sum of its singular values. 
Given a PSD matrix $\Ab \in \mathbb{R}^{n \times n}$, its Singular Value Decomposition is given by $\Ab = \Ub \Sigmab \Ub^T$, where $\Ub$ is the matrix of left/right singular vectors and $\Sigmab$ is the diagonal matrix of singular values.  

\subsection{Our Contributions}
Thresholding is a simple algorithmic concept, where each coordinate of, say, a vector is retained if its value is sufficiently high; otherwise, it is set to zero. Thresholding naturally preserves entries that have large magnitude while creating sparsity by eliminating small entries.
Thus, thresholding seems like a logical strategy for SPCA: after computing a dense vector that approximately solves a PCA problem, perhaps with additional constraints, thresholding can be used to sparsify it.

Our first approach (\texttt{spca-svd}, \secref{sec:pcp}) is a simple thresholding-based algorithm for SPCA that leverages the fact that the top singular vector is an optimal solution for the SPCA problem without the sparsity constraint. Our algorithm actually uses a thresholding scheme that leverages the top few singular vectors of the underlying covariance matrix; it  is simple and intuitive, yet offers the first of its kind runtime vs. accuracy bounds. Our algorithm returns a vector that is provably sparse and, when applied to the input covariance matrix $\Ab$, provably captures the optimal solution $\mathcal{Z}^*$ up to a small additive error. 
Indeed, our output vector has a sparsity that depends on $k$ (the target sparsity of the original SPCA problem of eqn.~(\ref{eqn:spca})) and $\varepsilon$ (an accuracy parameter between zero and one). Our analysis provides unconditional guarantees for the accuracy of the solution of the proposed  thresholding scheme. To the best of our knowledge, no such analyses have appeared in prior work (see \secref{sxn:prior} for details). 
We emphasize that our approach only requires an approximate singular value decomposition and, as a result, \texttt{spca-svd} runs very quickly. In practice, \texttt{spca-svd} is faster than current state-of-the-art and almost as accurate, at least in the datasets that we used in our empirical evaluations. 
However, as shown in \secref{sec:exp}, there is a clear theory-practice gap, since our theoretical bounds fail to predict the practical performance of \texttt{spca-svd}, which motivated us to look for more elaborate thresholding schemes that come with improved theoretical accuracy guarantees.

Our second approach (\texttt{spca-sdp}, \secref{sxn:SDPrelax}) uses a more elaborate semidefinite programming (SDP) approach with (relaxed) sparsity constraints to compute a starting point on which thresholding strategies are applied. Our algorithm provides novel bounds for the following standard convex relaxation of the problem of eqn.~\eqref{eqn:spca}:
\begin{equation}\label{eqn:spdrelax}
\max_{\Zb \in \mathbb{R}^{n \times n},\ \Zb \succeq 0}\ \trace(\A\Z)\ \, \text{s.t.}\ \trace(\Z)\le 1\ \text{and} \ \sum|Z_{i,j}|\le k.
\end{equation}
It is well-known that the optimal solution to eqn.~(\ref{eqn:spdrelax}) is greater than or equal to the optimal solution to eqn.~(\ref{eqn:spca}). 
We contribute a novel, two-step deterministic thresholding scheme that converts $\Zb \in \mathbb{R}^{n \times n}$ to a vector $\zb\in \mathbb{R}^n$ with sparsity $\O{\nicefrac{\beta^2k^2}{\eps^2}}$ and satisfies\footnote{For simplicity of presentation and following the lines of~\citep{FountoulakisKKD17}, we assume that the rows and columns of the matrix $\Ab$ have unit norm; this assumption can be removed as in~\citep{FountoulakisKKD17}.} $\z^\top\A\z\ge\frac{1}{\alpha}\cdot \mathcal{Z}^*-\eps$.
Here $\alpha$ and $\beta$ are parameters of the optimal solution matrix $\Zb$ that describe the extent to which the SDP relaxation of eqn.~(\ref{eqn:spdrelax}) is able to capture the original problem. We empirically demonstrate that these quantities are close to one for the (diverse) datasets used in our empirical evaluation. As a result (see~\secref{sec:exp}), we demonstrate that the empirical performance of \texttt{spca-sdp} is much better predicted by our theory, unlike \texttt{spca-svd} and current state-of-the-art. To the best of our knowledge, this is the first analysis of a rounding scheme for the convex relaxation of eqn.~(\ref{eqn:spdrelax}) that does not assume a specific model for the covariance matrix $\Ab$. However, this approach introduces a major runtime bottleneck for our algorithm, namely solving an SDP.

An additional contribution of our work is that, unlike prior work, our algorithms have clear tradeoffs between quality of approximation and output sparsity. Indeed, by increasing the density of the final SPCA vector, one can improve the amount of variance that is captured by our SPCA output. See~\thmref{thm:pcpmain} and~\thmref{thm:spca:sdp:det} for details on this sparsity vs. accuracy tradeoff for \texttt{spca-svd} and \texttt{spca-sdp}, respectively.


%
%

\textbf{Applications to Sparse Kernel PCA.}
Our algorithms have immediate applications to sparse kernel PCA (SKPCA), where the input matrix $\Ab\in\mathbb{R}^{n\times n}$ is instead implicitly given as a kernel matrix whose entry $(i,j)$ is the value $k(i,j):=\langle\phi(\X_{i*}),\phi(\X_{j*})\rangle$ for some kernel function $\phi$ that implicitly maps an observation vector into some high-dimensional feature space. 
Although $\Ab$ is not explicit, we can query all $\O{n^2}$ entries of $\Ab$ using $\O{n^2}$ time assuming an oracle that computes the kernel function $k$. 
We can then subsequently apply our SPCA algorithms and achieve polynomial runtime with the same approximation guarantees. 

\textbf{Experiments.}
Finally, we evaluate our algorithms on a variety of real and synthetic datasets in order to practically assess their performance. As discussed earlier, from an accuracy perspective, our algorithms perform comparably to current state-of-the-art. However, \texttt{spca-svd} is faster than current state-of-the-art. Importantly, we evaluate the tightness of the theoretical bounds on the approximation accuracy: the theoretical bounds for \texttt{spca-svd} and current state-of-the-art fail to predict the approximation accuracy of the respective algorithms in practice. However, the theoretical bounds for \texttt{spca-sdp} are much tighter, essentially bridging the theory-practice gap, at least in the datasets used in our evaluations. 

\subsection{Prior work}
\seclab{sxn:prior}
SPCA was formally introduced by~\citep{AGJ2007}; however, previously studied PCA approaches based on rotating~\citep{jolliffe1995rotation} or thresholding~\citep{CJ95} the top singular vector of the input matrix seemed to work well, at least in practice, given sparsity constraints. 
Following~\citep{AGJ2007}, there has been an abundance of interest in SPCA. 
\citep{jolliffe2003modified} considered LASSO (SCoTLASS) on an $\ell_1$ relaxation of the problem, while~\citep{zou2003regression} considered a non-convex regression-type approximation, penalized similar to LASSO.
Additional heuristics based on LASSO~\citep{ando2009approximating} and non-convex $\ell_1$ regularizations~\citep{zou2003regression,zou2006sparse,sriperumbudur2007sparse,shen2008sparse} have also been explored. 
Random sampling approaches based on non-convex $\ell_1$ relaxations~\citep{FountoulakisKKD17} have also been studied; we highlight that unlike our approach,~\citep{FountoulakisKKD17} solved a non-convex relaxation of the SPCA problem and thus perhaps relied on locally optimal solutions.  
\cite{beck2016sparse} presented a coordinate-wise optimality condition for SPCA and designed algorithms to find points satisfying that condition. While this method is guaranteed to converge to stationary points, it is also susceptible to getting trapped at local solutions. 
In another recent work \cite{Yuan_2019_CVPR}, the authors came up with a decomposition algorithm to solve the sparse generalized eigenvalue problem using a random or a swapping strategy. 
However, the underlying method  needs to solve a subproblem globally using combinatorial search methods at each iteration, which may fail for a large sparsity parameter $k$.
Additionally,~\citep{moghaddam2006spectral} considered a branch-and-bound heuristic motivated by greedy spectral ideas. 
\citep{journee2010generalized,PDK2013,Kuleshov13,yuan2013truncated} further explored other spectral approaches based on iterative methods similar to the power method. 
\citep{yuan2013truncated} specifically designed a sparse PCA algorithm with early stopping for the power method, based on the target sparsity. 

Another line of work focused on using semidefinite programming (SDP) relaxations~\citep{AGJ2007,d2008optimal,Amini09,dOrsiKNS20}. 
Notably,~\citep{Amini09} achieved provable theoretical guarantees regarding the SDP and thresholding approach of~\citep{AGJ2007} in a \textit{specific}, high-dimensional spiked covariance model, in which a base matrix is perturbed by adding a sparse maximal eigenvector. 
In other words, the input matrix is the identity matrix plus a ``spike'', i.e., a sparse rank-one matrix. 

Despite the variety of heuristic-based sparse PCA approaches, very few theoretical guarantees have been provided for SPCA; this is partially explained by a line of hardness-of-approximation results. 
The sparse PCA problem is well-known to be $\mathsf{NP}$-Hard~\citep{MWA2006}. 
\citep{magdon2017np} shows that if the input matrix is not PSD, then even the \emph{sign} of the optimal value cannot be determined in polynomial time unless $\mathsf{P=NP}$, ruling out any multiplicative approximation algorithm. 
In the case where the input matrix is PSD,~\citep{ChanPR16} shows that it is $\mathsf{NP}$-hard to approximate the optimal value up to multiplicative $(1+\varepsilon)$ error, ruling out any polynomial-time approximation scheme (PTAS). 
Moreover, they show Small-Set Expansion hardness for any polynomial-time constant factor approximation algorithm and also that the standard SDP relaxation might have an exponential gap.

We conclude by summarizing prior work that offers provable guarantees (beyond the work of~\citep{Amini09}), typically given \textit{some assumptions about the input matrix}.~\citep{d2014approximation} showed that the SDP relaxation can be used to find provable bounds when the covariance input matrix is formed by a number of data points sampled from Gaussian models with a single sparse singular vector. Perhaps the most interesting, theoretically provable, prior work is~\citep{PDK2013}, which presented a combinatorial algorithm that analyzed a specific set of vectors in a low-dimensional eigenspace of the input matrix and presented relative error guarantees for the optimal objective, given the assumption that the input covariance matrix has a decaying spectrum. From a theoretical perspective, this method is the current state-of-the-art and we will present a detailed comparison with our approaches in \secref{sec:exp}.~\citep{asteris2011sparse} gave a polynomial-time algorithm that solves sparse PCA \textit{exactly} for input matrices of constant rank. 
~\citep{ChanPR16} showed that sparse PCA can be approximated in polynomial time within a factor of~$n^{-1/3}$ and also highlighted an additive PTAS of ~\citep{asteris2015sparse} based on the idea of finding multiple disjoint components and solving bipartite maximum weight matching problems. 
This PTAS needs time $n^{\poly(1/\eps)}$, whereas all of our algorithms have running times that are a low-degree polynomial in $n$.

\section{SPCA via SVD Thresholding}
\seclab{sec:pcp}
To achieve nearly input sparsity runtime, our first thresholding algorithm is based upon using the top $\ell$ right singular vectors of the PSD matrix $\Ab$. Given $\Ab$ and an accuracy parameter $\varepsilon$, our approach first computes $\Sigmab_{\ell} \in \mathbb{R}^{\ell \times \ell}$ (the diagonal matrix of the top $\ell$ singular values of $\Ab$) and $\Ub_{\ell} \in \mathbb{R}^{n \times \ell}$ (the matrix of the top $\ell$ right singular vectors of $\Ab$), for $\ell = 1/\varepsilon$. Then, it \textit{deterministically} selects a subset of $\O{k/\varepsilon^3}$ columns of $\Sigmab^{1/2}_{\ell}\Ub_{\ell}^\top$ using a simple thresholding scheme based on the norms of the columns of $\Sigmab^{1/2}_{\ell}\Ub_{\ell}^\top$. (Recall that $k$ is the sparsity parameter of the SPCA problem.)
In the last step, it returns the \textit{top right singular vector} of the matrix consisting of the chosen columns of $\Sigmab_{\ell}^{1/2}\Ub_{\ell}^\top$. Notice that this right singular vector is an $\O{k/\varepsilon^3}$-dimensional vector, which is finally expanded to a vector in $\mathbb{R}^n$ by appropriate padding with zeros. This sparse vector is our approximate solution to the SPCA problem of eqn.~(\ref{eqn:spca}).

This simple algorithm is somewhat reminiscent of prior thresholding approaches for SPCA. However, to the best of our knowledge, no provable a priori bounds were known for such algorithms without strong assumptions on the input matrix. This might be due to the fact that prior approaches focused on thresholding only the top right singular vector of $\Ab$, whereas our approach thresholds the top $\ell = 1/\varepsilon$ right singular vectors of $\Ab$. This slight relaxation allows us to present provable bounds for the proposed algorithm.

In more detail, let the SVD of $\Ab$ be $\Ab = \Ub \Sigmab \Ub^T$. Let $\Sigmab_{\ell} \in \mathbb{R}^{\ell \times \ell}$ be the diagonal matrix of the top $\ell$ singular values and let $\Ub_{\ell} \in \mathbb{R}^{n \times \ell}$ be the matrix of the top $\ell$ right (or left) singular vectors. Let $R = \{i_1,\ldots,i_{|R|}\}$ be the set of indices of rows of $\Ub_{\ell}$ that have squared norm at least $\varepsilon^2/k$ and let $\bar{R}$ be its complement. Here $|R|$ denotes the cardinality of the set $R$ and $R \cup \bar{R} = \{1,\ldots,n\}$. Let $\Rb \in \mathbb{R}^{n \times |R|}$ be a sampling matrix that selects\footnote{Each column of $\Rb$ has a single non-zero entry (set to one), corresponding to one of the $|R|$ selected columns. Formally, $\Rb_{i_t,t}=1$ for $t=1,\ldots, |R|$; all other entries of $\Rb$ are set to zero.}
the columns of $\Ub_{\ell}$ whose indices are in the set $R$. Given this notation, we are now ready to state~\algref{alg:sparse:pca}.
\begin{algorithm}[!htb]
\caption{\texttt{spca-svd}: fast thresholding SPCA via SVD}
\alglab{alg:sparse:pca}
\begin{algorithmic}[1]
\Require{$\A\in\mathbb{R}^{n\times n}$, sparsity $k$, error parameter $\eps>0$.}
\Ensure{$\y \in \mathbb{R}^n$ such that $\|y\|_2=1$ and $\|\y\|_0 = k/\varepsilon^2$.}
\State{$\ell \gets 1/\varepsilon$;}
\State{Compute $\Ub_{\ell} \in \mathbb{R}^{n \times \ell}$ (top $\ell$ left singular vectors of $\Ab$) and $\Sigmab_{\ell} \in \mathbb{R}^{\ell \times \ell}$
(the top $\ell$ singular values of $\Ab$)};
\State{Let $R = \{i_1,\ldots, i_{|R|}\}$ be the set of rows of $\Ub_{\ell}$ with squared norm at least $\eps^2/k$ and let $\Rb \in \mathbb{R}^{n \times |R|}$ be the associated sampling matrix (see text for details);}
\State{$\y \in \mathbb{R}^{|R|}\gets\argmax_{\norm{\x}_2=1}\norm{\Sigmab_{\ell}^{1/2}\Ub_\ell^\top\Rb\x}_2^2$;}
\State{\textbf{return} $\zb = \Rb\y \in \mathbb{R}^n$;}
\end{algorithmic}
\end{algorithm}

Notice that $\Rb\y$ satisfies $\|\Rb\y\|_2=\|\y\|_2=1$ (since $\Rb$ has orthogonal columns) and $\|\Rb\y\|_0=|R|$. Since $R$ is the set of rows of $\Ub_{\ell}$ with squared norm at least $\eps^2/k$ and $\|\Ub_{\ell}\|_F^2 = \ell = 1/\varepsilon$, it follows that $|R|\leq k/\varepsilon^3$. 
Thus, the vector returned by~\algref{alg:sparse:pca} has $k/\varepsilon^3$ sparsity and unit norm. 
\begin{restatable}{theorem}{thmpcpmain}
\thmlab{thm:pcpmain}
Let $k$ be the sparsity parameter and $\eps \in (0,1]$ be the accuracy parameter. Then, the vector $\zb\in\mathbb{R}^{n}$ (the output of~\algref{alg:sparse:pca}) has sparsity $k/\varepsilon^3$, unit norm, and satisfies 
$$\zb^\top \Ab \zb \geq \mathcal{Z}^* - 3\eps\trace(\A).$$
\end{restatable}

The intuition behind \thmref{thm:pcpmain} is that we can decompose the value of the optimal solution into the value contributed by the coordinates in $R$, the value contributed by the coordinates outside of $R$, and a cross term. 
The first term we can upper bound by the output of the algorithm, which maximizes with respect to the coordinates in $R$. 
For the latter two terms, we can upper bound the contribution due to the upper bound on the squared row norms of indices outside of $R$ and due to the largest singular value of $\Ub$ being at most the trace of $\Ab$. 
We defer the full proof of \thmref{thm:pcpmain} to the supplementary material. 

The running time of~\algref{alg:sparse:pca} is dominated by the computation of the top $\ell$ singular vectors and singular values of the matrix $\Ab$. 
One could always use the SVD of the full matrix $\Ab$ ($\O{n^3}$ time) to compute the top $\ell$ singular vectors and singular values of $\Ab$. 
In practice, any iterative method, such as subspace iteration using a random initial subspace or the Krylov subspace of the matrix, can be used towards this end. 
We address the inevitable approximation error incurred by such approximate SVD methods below. 

%
Finally, we highlight that, as an intermediate step in the proof of \thmref{thm:pcpmain}, we need to prove the following \lemref{lem:sz:pcp}, which is very much at the heart of our proof of \thmref{thm:pcpmain} and, unlike prior work, allows us to provide provably accurate bounds for the thresholding~\algref{alg:sparse:pca}.
\begin{restatable}{lemma}{lemszpcp}
\lemlab{lem:sz:pcp}
Let $\A\in\mathbb{R}^{n\times n}$ be a PSD matrix and $\Sigmab \in \mathbb{R}^{n \times n}$ (respectively,
$\Sigmab_{\ell} \in \mathbb{R}^{\ell \times \ell}$) be the diagonal matrix of all (respectively, top $\ell$) singular values and let
$\Ub \in \mathbb{R}^{n \times n}$ (respectively, $\Ub_{\ell} \in \mathbb{R}^{n \times \ell}$) be the matrix of all (respectively, top $\ell$) singular vectors. Then, for all unit vectors $\xb \in \mathbb{R}^n$,
$$\norm{\Sigmab_{\ell}^{1/2} \Ub_{\ell}^\top\x}_2^2 \ge \norm{\Sigmab^{1/2} \Ub^\top \x}_2^2- \eps\trace(\A).$$
\end{restatable}
%
At a high level, the proof of \lemref{lem:sz:pcp} first decomposes a basis for the columns spanned by $\Ub$ into those spanned by the top $\ell$ singular vectors and the remaining $n-\ell$ singular vectors. 
We then lower bound the contribution of the top $\ell$ singular vectors by upper bounding the contribution of the remaining $n-\ell$ singular vectors after noting that the largest remaining singular value is at most a $1/\ell$-fraction of the trace. 
For additional details, we defer the full proof to the supplementary material.

\paragraph{Using an approximate SVD solution.} 
The guarantees of \thmref{thm:pcpmain} in \algref{alg:sparse:pca} use an exact SVD computation, which could take time $\O{n^3}$. 
We can further improve the running time by using an approximate SVD algorithm such as the randomized block Krylov method of~\cite{MuscoM15}, which runs in nearly input sparsity runtime. 
Our analysis uses the relationships $\norm{\Sigmab_{\ell,\perp}^{1/2}}_2^2 \leq \nicefrac{\trace(\Ab)}{\ell}$ and $\sigma_1(\Sigma_{\ell})\le\trace(\Ab)$. 
The randomized block Krylov method of~\cite{MuscoM15} recovers these guarantees up to a multiplicative $(1+\eps)$ factor, in $\O{\nicefrac{\log n}{\eps^{1/2}}\cdot\nnz(\A)}$ time. 
Thus, by rescaling $\eps$, we recover the same guarantees of \thmref{thm:pcpmain} by using an approximate SVD in nearly input sparsity time. This results in a randomized algorithm; if one wants a deterministic algorithm then one should compute an exact SVD.   

\section{SPCA via SDP Relaxation and Thresholding}\seclab{sxn:SDPrelax}
To achieve higher accuracy, our second thresholding algorithm uses an approach that is based on the SDP relaxation of eqn.~(\ref{eqn:spdrelax}). Recall that solving eqn.~(\ref{eqn:spdrelax}) returns a PSD matrix $\Zb^* \in \mathbb{R}^{n \times n}$ that, by the definition of the semidefinite programming relaxation, satisfies
$\trace(\Ab\Zb^*) \geq \mathcal{Z}^*$, where $\mathcal{Z}^*$ is the true optimal solution of SPCA in eqn.~(\ref{eqn:spdrelax}). 

We would like to acquire a sparse vector $\zb$ from $\Z^*$. 
To that end, we first take $\Z_1$ to be the best rank-$1$ approximation to $\Z^*$. 
Note that $\Z_1$ can be quickly computed by taking the top eigenvector $\ub$ of $\Z^*$ and setting $\Z_1=\ub\ub^\top$. 
Consider a set $S$ defined to be the set of indices of the $9k^2\beta^2/\eps^2$ coordinates of $\ub$ with the largest absolute value, where $\beta$ is a parameter defined below (and close to $1$ in our experiments). 
Intuitively, the indices in $S$ should correlate with the ``important'' rows and columns of $\Ab$. 
Hence, we define $\zb\in\mathbb{R}^n$ to be the vector that matches the corresponding coordinates of $\ub$ for indices of $S$ and zero elsewhere, outside of $S$. 
As a result, $\zb$ will be a $9k^2\beta^2/\eps^2$-sparse vector.  

\begin{algorithm}[!htb]
\caption{\texttt{spca-sdp}: accurate thresholding SPCA via SDP}
\alglab{alg:sparse:sdp:det}
\begin{algorithmic}[1]
\Require{$\A\in\mathbb{R}^{n\times n}$, sparsity $k$, error parameter $\eps>0$.}
\Ensure{$\zb \in \mathbb{R}^n$ such that $\|\zb\|_2=1$ and $\|\zb\|_0 = 9k^2\beta^2/\eps^2$.}
\State{Let $\Z^*$ be the optimal solution to the relaxed SPCA problem of eqn.~(\ref{eqn:spdrelax});}
\State{Let $\Z_1=\ub \ub^\top$ be the best rank-$1$ approximation to $\Z^*$;}
\State{Let $\zb\in\mathbb{R}^n$ be the sparse vector containing the top $\nicefrac{9k^2 \beta^2}{\varepsilon^2}$ coordinates in magnitude of $\ub$;}
\State{\textbf{return} $\zb$;}
\end{algorithmic}
\end{algorithm}

Our main quality-of-approximation result for~\algref{alg:sparse:sdp:det} is \thmref{thm:spca:sdp:det}. 
For simplicity of presentation, we make the standard assumption that all rows and columns of $\Ab$ have been normalized to have unit norm; this assumption can be relaxed, e.g., see~\citep{FountoulakisKKD17}. 

\begin{theorem}
\thmlab{thm:spca:sdp:det}
Given a PSD matrix $\A\in\mathbb{R}^{n\times n}$, a sparsity parameter $k$, and an error tolerance $\eps>0$, let $\Z$ be an optimal solution to the relaxed SPCA problem of eqn.~(\ref{eqn:spdrelax}) and $\Zb_1$ be the best rank-one approximation to $\Zb$. 
Suppose $\alpha\ge1$ is a constant such that $\trace(\Ab\Zb)\le \alpha\trace(\Ab\Zb_1)$ and $\beta$ is a constant such that $\beta\ge\frac{\|\Zb_1\|_1}{\|\Zb\|_1}$. 
Then,~\algref{alg:sparse:sdp:det} outputs a vector $\zb \in \mathbb{R}^n$ that satisfies $\norm{\z}_0=\frac{9k^2\beta^2}{\varepsilon^2}$,  $\norm{\z}_2\le 1$, and
%
\[\z^\top\A\z\ge(1/\alpha)\mathcal{Z}^*-\eps.\]
%
%
%
\end{theorem}
\begin{proof}
If $\Zb_1 = \ub\ub^\top$, then $\ub^\top\Ab \ub = \trace(\Ab\Zb_1) \ge (1/\alpha) \trace(\Ab\Zb) \ge (1/\alpha)\mathcal{Z}^*$.

Since all eigenvalues of $\Zb$ are at most one, we have that $\|\ub\ub^\top\|_F^2 \le 1$, which means $\|\ub\|_2 \le 1$.

Also, letting $\beta \ge \frac{\|\Zb_1\|_1}{\|\Zb\|_1}$, we have $\|\Zb_1\|_1 \le \beta \|\Zb\|_1 \le \beta k$. So  $\|\Zb_1\|_1 = \sum_{i,j} \abs{u_i u_j} = \|\ub\|_1^2$, so $\|\ub\|_1 \le \sqrt{\beta k}$.

Let $\zb$ be the vector of top $\frac{9k^2 \beta^2}{\varepsilon^2}$ coordinates in absolute value of $\ub$, and remaining entries equal to $0$. 
Then by H\"{o}lder's inequality,
\begin{flalign*}
\zb^\top \Ab\zb =~& \ub^\top \Ab\ub - (\ub-\zb)^\top \Ab\zb - \zb^\top \Ab (\ub-\zb) - (\ub-\zb)^\top \Ab (\ub-\zb)\nonumber\\
=~& \ub^\top \Ab\ub - 2\zb^\top \Ab (\ub-\zb) - (\ub-\zb)^\top \Ab (\ub-\zb)^\top\nonumber\\
\ge~& \ub^\top \Ab\ub - 2 \|\zb\|_1 \|\Ab(\ub-\zb)\|_\infty - \|\ub-\zb\|_1 \|\Ab(\ub-\zb)\|_\infty.
\end{flalign*}
Each row of $\Ab$ has squared Euclidean norm at most $1$, so that $\|\ub-\zb\|_2\ge\|\Ab(\ub-\zb)\|_\infty$. 
Since $\|\ub\|_1 \le \sqrt{\beta k}$, then
\[
\zb^\top \Ab\zb\ge\ub^\top \Ab\ub - 2 \sqrt{\beta k} \|\ub-\zb\|_2 - \sqrt{\beta k} \|\ub-\zb\|_2= \ub^\top \Ab\ub - 3 \sqrt{\beta k} \|\ub-\zb\|_2.\]


Because $\|\ub\|_1 \le \sqrt{\beta k}$ and $\zb$ contains the top $\frac{9k^2 \beta^2}{\varepsilon^2}$ coordinates in absolute value of $\ub$, then all entries of $(\ub-\zb)$ have magnitude at most $\frac{\varepsilon^2}{9(k \beta)^{3/2}}$. 
Hence, we have
\[\|\ub-\zb\|_2^2 \le \frac{\varepsilon^4}{81(k\beta)^3} \cdot\frac{9(k\beta)^2}{\varepsilon^2}=\frac{\varepsilon^2}{9k \beta},
\]
and therefore, $\|\ub-\zb\|_2 \le \frac{\varepsilon}{3\sqrt{k \beta}}$. 
Hence, 
$$\zb^\top \Ab\zb \ge \ub^\top \Ab\ub - \varepsilon \ge (1/\alpha) \mathcal{Z}^* - \varepsilon.$$
Note $\zb$ has $\frac{9k^2 \beta^2}{\varepsilon^2}$ non-zero entries, and $\|\zb\|_2 \le \|\ub\|_2\le 1$.
\end{proof}

\paragraph{Interpretation of our guarantee.} Our assumptions in \thmref{thm:spca:sdp:det} simply say that much of the trace of the matrix $\Ab\Zb$ should be captured by the trace of $\Ab\Zb_1$, as quantified by the constant $\alpha$. For example, if $\Zb$ were a rank-one matrix, then the assumption would hold with $\alpha=1$. As the trace of $\Ab\Zb_1$ fails to approximate the trace of $\Ab\Zb$ (which intuitively implies that the SDP relaxation of eqn.~\eqref{eqn:spdrelax} did not sufficiently capture the original problem), the constant $\alpha$ increases and the quality of the approximation decreases. In our experiments, we indeed observed that $\alpha$ is close to one (see Table~\ref{tab:constants}). Similarly, we empirically observe that $\beta$, which is the ratio between the 1-norm of $\Zb$ and its \emph{rank-one} approximation is also close to one for our datasets.

\paragraph{Using an approximate SDP solution.} 
The guarantees of \thmref{thm:spca:sdp:det} in \algref{alg:sparse:sdp:det} use an optimal solution $\Zb$ to the SDP relaxation in eqn.~(\ref{eqn:spdrelax}). 
In practice, we will only obtain an approximate solution $\tilde{\Zb}$ to eqn.~(\ref{eqn:spdrelax}) using any standard SDP solver, e.g.~\citep{Ali1995}, such that $\trace(\Ab\tilde{\Zb})\ge\trace(\Ab\Zb)-\varepsilon$ after $\O{\log\nicefrac{1}{\varepsilon}}$ iterations. 
Since our analysis only uses the relationship $\ub^\top\A\ub\ge\trace(\A\Z)$, then the additive $\varepsilon$ guarantee can be absorbed into the other $\varepsilon$ factors in the guarantees of \thmref{thm:spca:sdp:det}. 
Thus, we recover the same guarantees of \thmref{thm:spca:sdp:det} by using an approximate solution to the SDP relaxation in eqn.~(\ref{eqn:spdrelax}). 


\section{Experiments}
\seclab{sec:exp}

We compare the output of our algorithms against state-of-the-art SPCA approaches, including the coordinate-wise optimization algorithm of~\cite{beck2016sparse}~(\texttt{cwpca}) the block decomposition algorithm of~\cite{Yuan_2019_CVPR}~(\texttt{dec}), and the spannogram-based algorithm of~\cite{PDK2013}~(\texttt{spca-lowrank}). 
For the implementation of \texttt{dec}, we used the \emph{coordinate descent method} 
and for \texttt{cwpca} we used the \emph{greedy coordinate-wise} (GCW) method. We implemented \texttt{spca-lowrank} with the low-rank parameter $d$ set to three; finally, for \texttt{spca-svd}, we fixed the threshold parameter $\ell$ to one.

\begin{figure*}[!htb]
\begin{subfigure}{0.33\textwidth}
  \centering
  \includegraphics[width=1\linewidth]{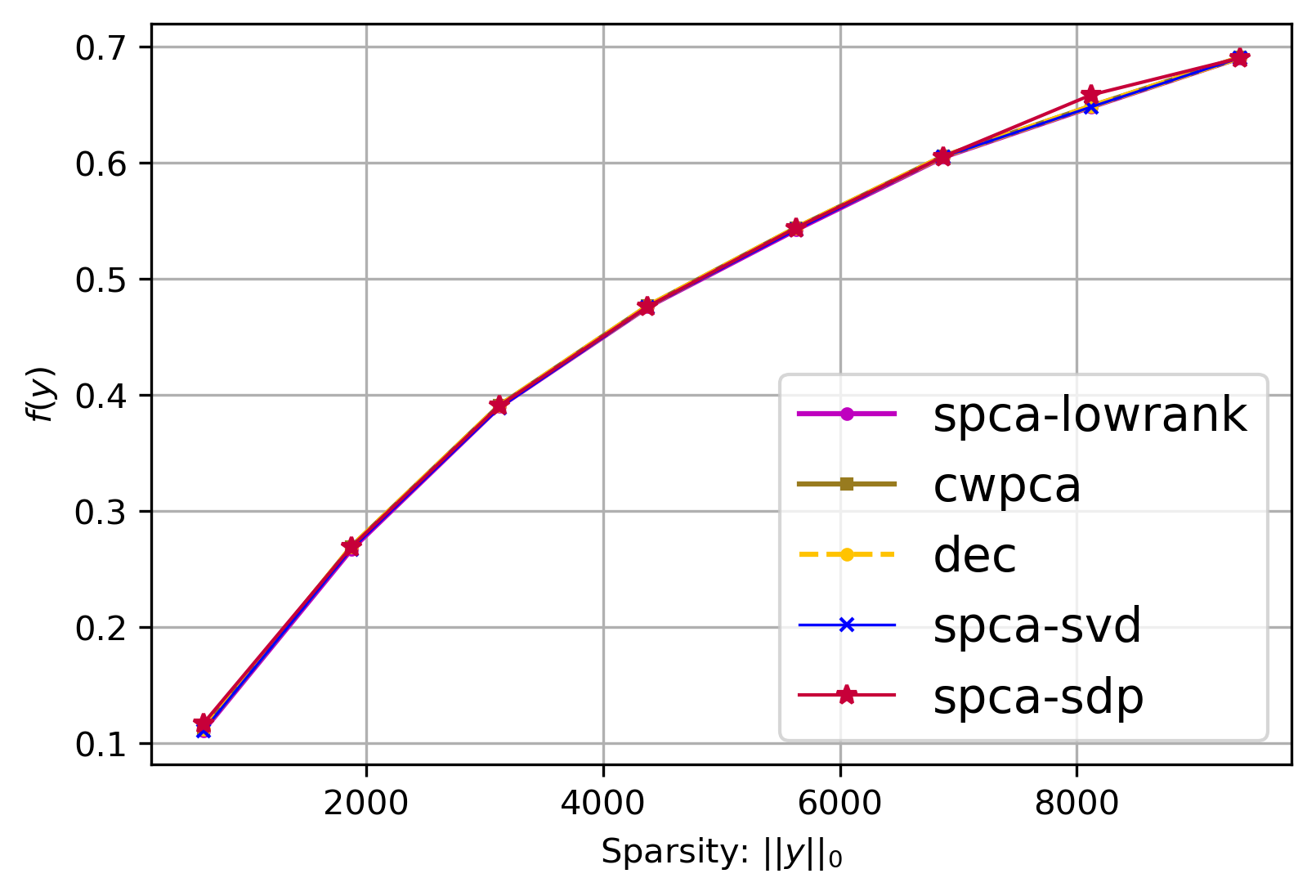}
  \caption{\textsc{Chr}~1, $n=37,493$}
  \label{fig:sfig1}
\end{subfigure}%
\begin{subfigure}{0.33\textwidth}
  \centering
  \includegraphics[width=1\linewidth]{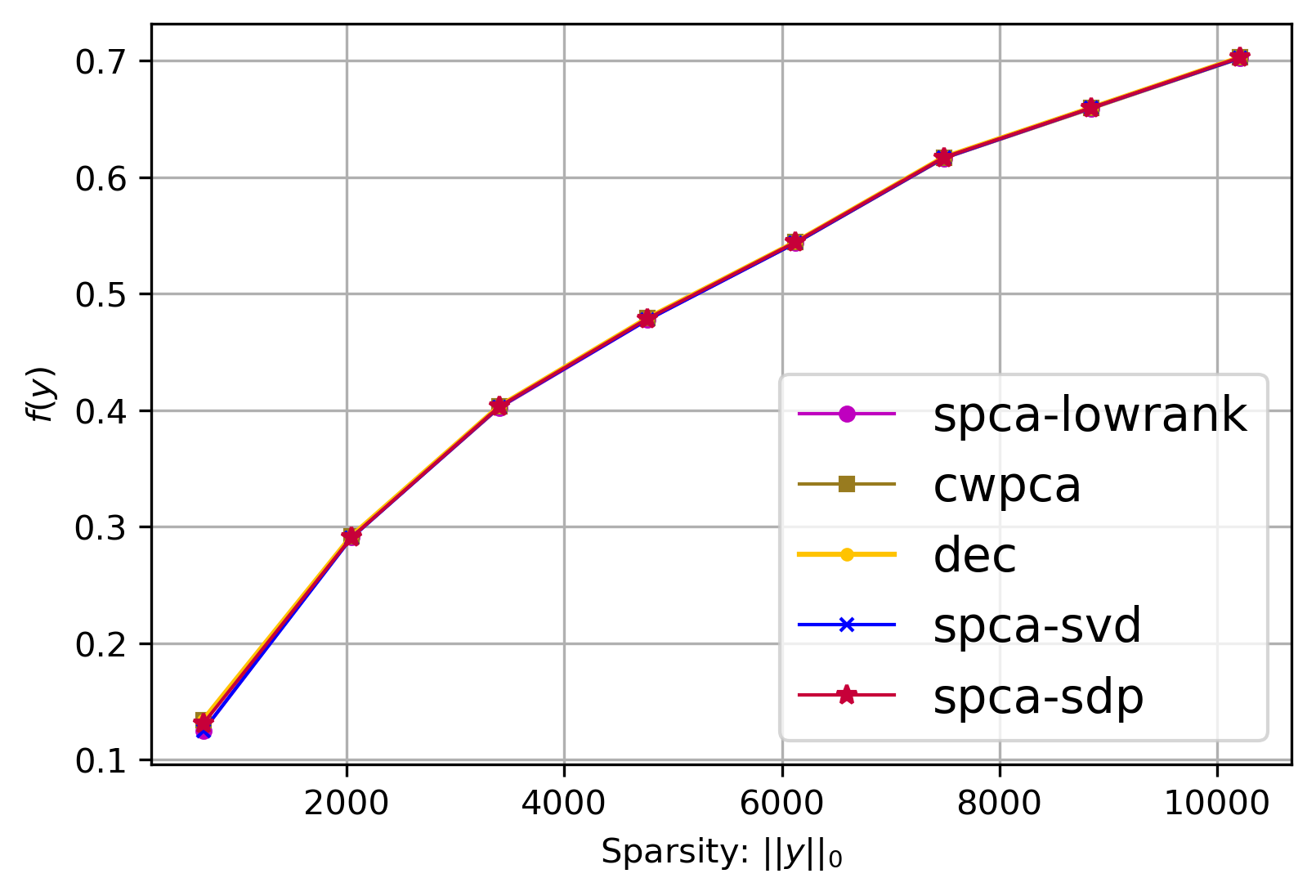}
  \caption{\textsc{Chr}~2, $n=40,844$}
  \label{fig:sfig2}
\end{subfigure}
\begin{subfigure}{0.33\textwidth}
\centering
    \includegraphics[width=1\linewidth]{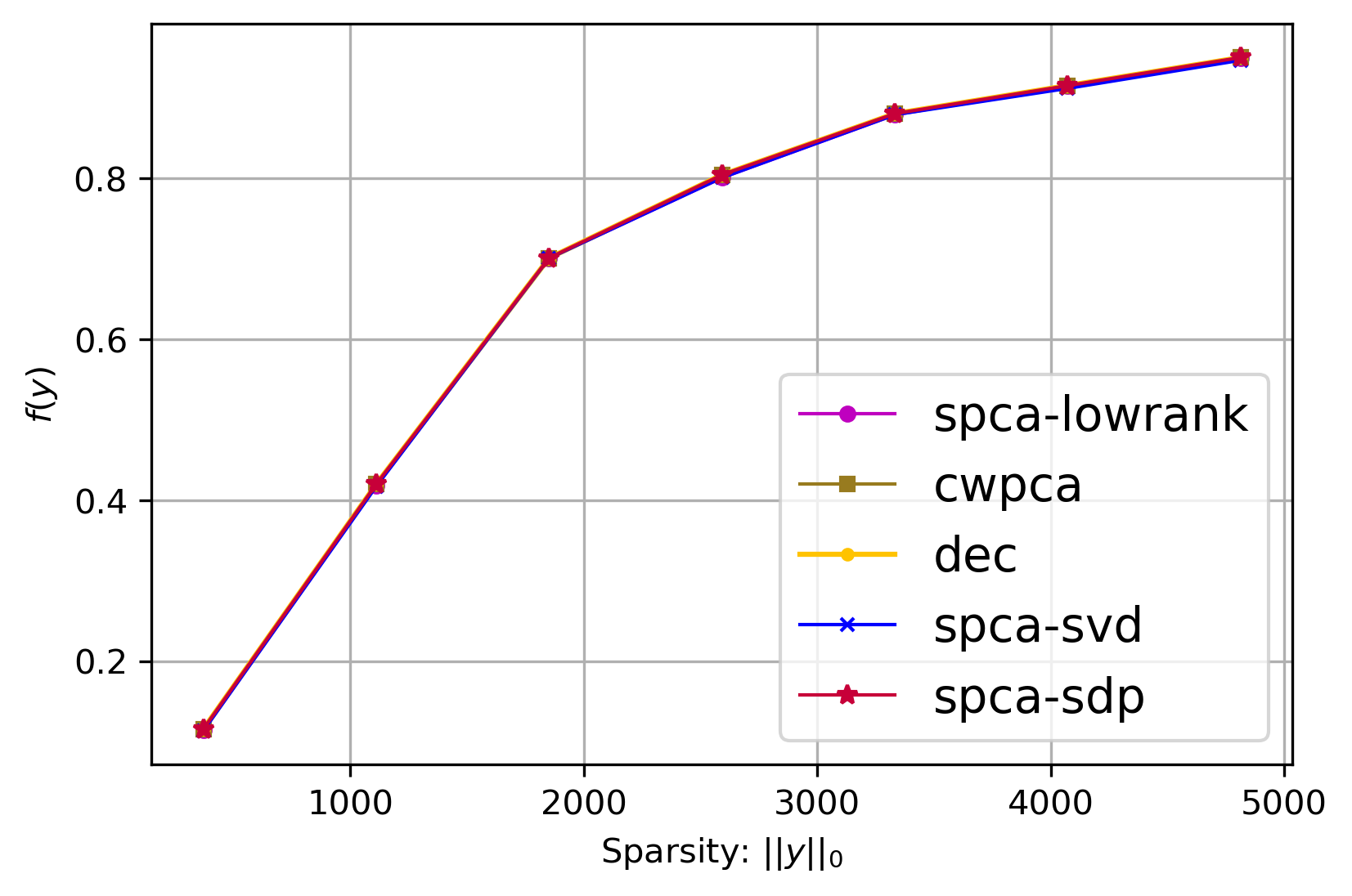}
  \caption{Gene expression data}
  \label{fig:sfig33}
\end{subfigure}

\caption{Experimental results on real data: $f(\y)$ vs. sparsity.}
\label{fig:fig}
\end{figure*}

In order to explore the sparsity patterns of the outputs, we first applied our methods on the \emph{pit props} dataset, which was introduced in~\citep{jeffers1967two} and is a toy, benchmark example used to test sparse PCA. It is a $13\times13$ correlation matrix, originally calculated from $180$ observations with $13$ explanatory variables. We applied our algorithms to the \emph{Pit Props} matrix in order to extract a sparse top principal component, having a sparsity pattern similar to that of 
\texttt{cwpca}, \texttt{dec}, and \texttt{spca-lowrank}. It is actually known that the decomposition method of~\cite{Yuan_2019_CVPR} can find the global optimum for this dataset. 
We set the sparsity parameter $k$ to seven; Table~\ref{tab:multicol1} in Appendix~\ref{app:expt} shows that both \texttt{spca-svd} and \texttt{spca-sdp} are able to capture the right sparsity pattern. In terms of the optimal value $\mathcal{Z}^*$, \texttt{spca-svd} performs very similar to \texttt{spca-lowrank}, while our SDP-based algorithm
\texttt{spca-sdp} exactly recovers the optimal solution and matches both \texttt{dec} and \texttt{cwpca}. 

\begin{figure}
    \centering

\begin{minipage}{0.32\linewidth}
  \includegraphics[height=1.4in,width=1.90in]{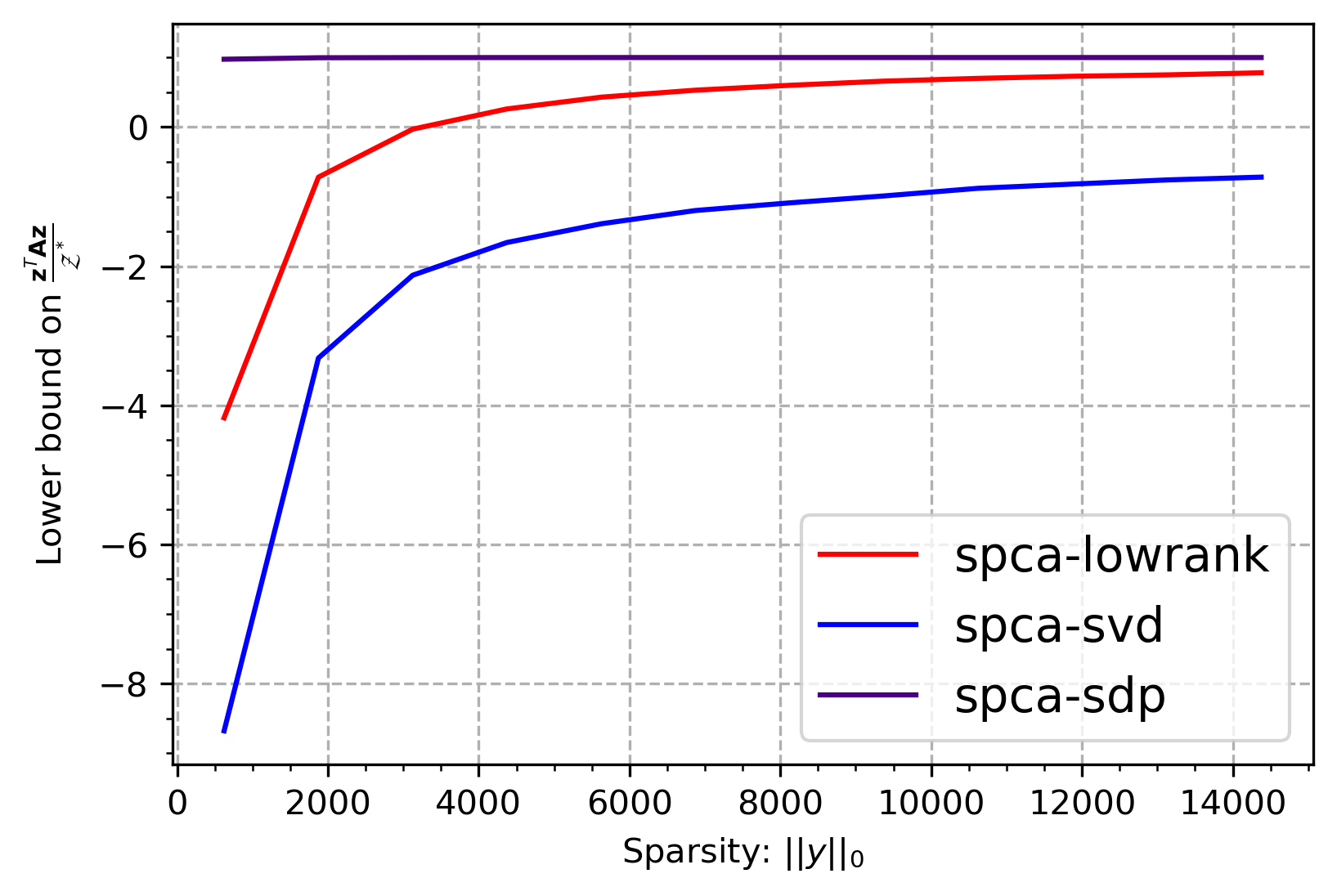}

\end{minipage}
\hspace{0.8mm}
\begin{minipage}{0.32\linewidth}
  \includegraphics[height=1.4in,width=1.90in]{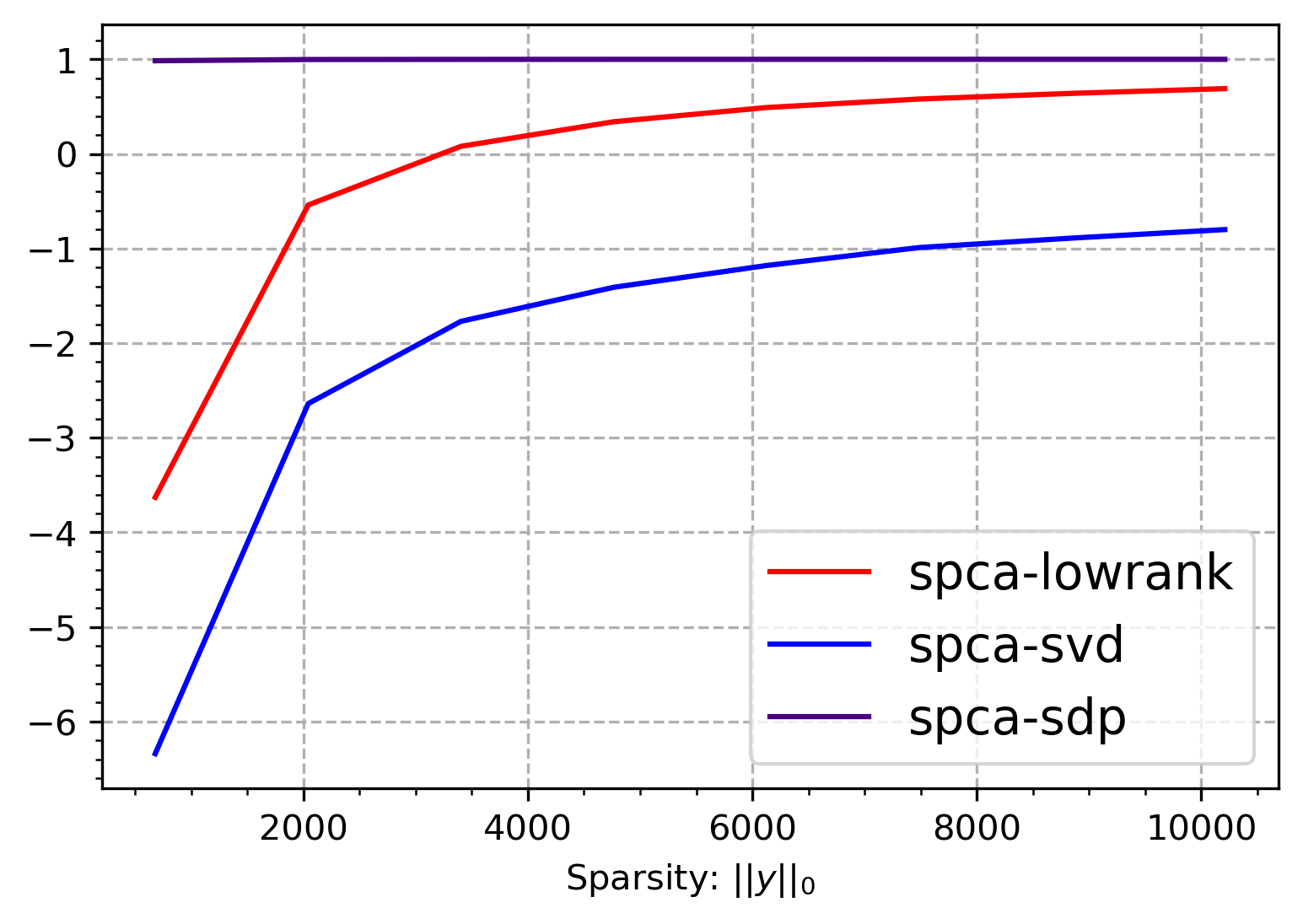}
\end{minipage}
\hspace{1.2mm}
\begin{minipage}{0.32\linewidth}
  \includegraphics[height=1.4in,width=1.90in]{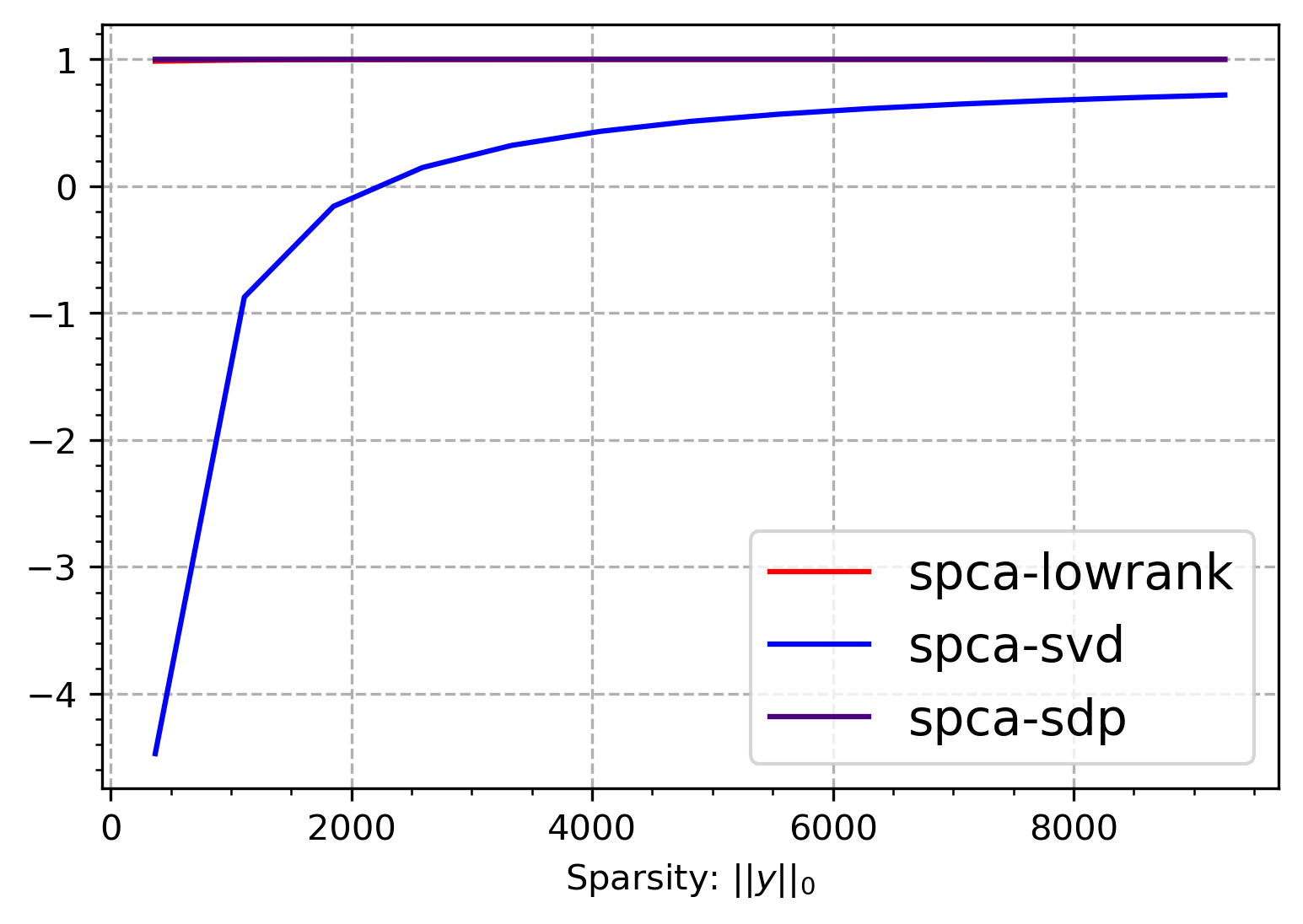}
\end{minipage}

    \caption{Tightness of our bounds: HGDP/HAPMAP chromosome 1 (left), HGDP/HAPMAP chromosome 2 (middle), and gene expression (right) data}
    \label{fig:my_label}
\end{figure}

Next, we further demonstrate the empirical performance of our algorithms on larger real-world datasets, as well as on synthetic datasets, similar to~\citep{FountoulakisKKD17}\,(see Appendix~\ref{app:expt}). We use genotypic data from the Human Genome Diversity Panel (HGDP)~\citep{int07} and the International Haplotype Map (HAPMAP) project~\citep{li2008}, forming 22 matrices, one for each chromosome, encoding all autosomal genotypes. 
Each matrix contains 2,240 rows and a varying number of columns (typically in the tens of thousands) that is equal to the number of single nucleotide polymorphisms (SNPs, well-known biallelic loci of genetic variation across the human genome) in the respective chromosome. Finally, we also use a lung cancer gene expression dataset ($107\times 22,215$ matrix) from~\citep{landi2008gene}.

We compare our algorithms \texttt{spca-svd} and~\texttt{spca-sdp} with the solutions returned by \texttt{dec}, \texttt{cwpca}, and \texttt{spca-lowrank}. Let $f(\y)=\nicefrac{\y^\top\Ab\y}{\|\Ab\|_2}$; then, $f(\y)$ measures the quality of an approximate solution $\y\in\mathbb{R}^n$ to the SPCA problem.  Essentially, $f(\yb)$ quantifies the ratio of the explained variance coming from the first sparse PC to the explained variance of the first \textit{sparse} eigenvector. Note that $0\le f(\y)\le 1$ for all $\y$ with $\|\y\|_2 \le 1$. As $f(\y)$ gets closer to one, the vector $\y$ captures more of the variance of the matrix $\Ab$ that corresponds to its top singular value and corresponding singular vector.

In our experiments, for \texttt{spca-svd} and \texttt{spca-sdp}, we fix the sparsity $s$ to be equal to $k$, so that all algorithms return a sparse vector with the same number of non-zero elements.  In Figures~\ref{fig:sfig1}-\ref{fig:sfig33} we evaluate the performance of the different SPCA algorithms by plotting $f(\y)$ against $\|\y\|_0$, i.e., the sparsity of the output vector, on data from HGDP/HAPMAP chromosome 1, HGDP/HAPMAP chromosome 2, and the gene expression data. 
Note that in terms of accuracy both \texttt{spca-svd}  and \texttt{spca-sdp}, essentially match the current state-of-the-art \texttt{dec}, \texttt{cwpca}, and \texttt{spca-lowrank}.

We now discuss the running time of the various approaches. All experiments were performed on a server with two Rome 32 core, $2.0$GHz CPUs and $.8$ TBs of RAM. Our \texttt{spca-svd} is the fastest approach, taking about 2.5 to three hours for each sparsity value for the largest datasets (HGDP/HAPMAP chromosomes~1 and~2 data). The state-of-the-art methods \texttt{spca-lowrank}, \texttt{cwpca}, and \texttt{dec} take five to seven hours for the same dataset and thus are at least two times slower. However, our second approach \texttt{spca-spd} is slower and takes approximately 20 hours, as it needs to solve a large-scale SDP problem. 

\textbf{Comparing the theoretical guarantees of our approaches and current-state-of-the-art.} A straightforward theoretical comparison between the theoretical guarantees of our methods and current state-of-the-art if quite tricky, since they depend on different parameters that are not immediately comparable. Therefore, we attempt to compare the \textit{tightness} of the theoretical guarantees in the context of the datasets used in our experiments, where these parameters can be directly evaluated. We chose to compare our theoretical bounds with the method of~\citep{PDK2013}, namely \texttt{spca-lowrank}, which works well in practice, has reasonable running times, and comes with state-of-the-art provable accuracy guarantees. We used the following datasets in our comparison: the data from chromosome 1 and chromosome 2 of HGDP/HapMap, and the gene expression data of~\cite{landi2008gene}. We set the accuracy parameter $\epsilon$ in~\thmref{thm:pcpmain} and~\thmref{thm:spca:sdp:det} to .9 (using different values of $\epsilon$ between zero and one does not change the findings of Figure~\ref{fig:my_label}). Figure~\ref{fig:my_label} summarizes our findings and highlights an important observation on the theory-practice gap: notice that for small values of the sparsity parameter, both \texttt{spca-lowrank} and \texttt{spca-svd} predict \textit{negative} values for the accuracy ratio, which are, of course, meaningless. The theoretical bounds of \texttt{spca-lowrank} become meaningful (e.g., non-negative) when the sparsity parameter exceeds 3.5K for the HAPMAP/HGDP chromosome 1 data, while the theoretical bounds of \texttt{spca-svd} remain consistently meaningless, despite its solid performance in practice. However, the theoretical bounds of our \texttt{spca-sdp} algorithm are consistently the best and very close to one, thus solidly predicting its high accuracy in practice. Similar findings are shown in the other panels of Figure~\ref{fig:my_label} for the other datasets (see Appendix~\ref{app:expt} for additional experiments), with the notable exception of the gene expression dataset, whose underlying eigenvalues nearly follow a power-law decay, and, as a result, \texttt{spca-lowrank} exhibits a much tighter bound that almost matches \texttt{spca-sdp}.
We believe that the improved theoretical performance of \texttt{spca-sdp} is due to the novel dependency of our approach (see \thmref{thm:spca:sdp:det}) on the constants $\alpha$ and $\beta$, that are both close to one in real datasets (see Table~\ref{tab:constants}). On the other hand, we note that the approximation guarantee of \texttt{spca-lowrank} typically depends on the the spectrum of $\Ab$ and the maximum diagonal entry of $\Ab$, which are less well-behaved quantities. 


\begin{table*}[ht]
\begin{center}
 \begin{minipage}{.47\textwidth}
 
		\resizebox{\textwidth}{!}{\begin{tabular}{|c| c c c|}
				\hline
				Sparsity& $\alpha$ & $\beta$ & $\epsilon_d$ \\
				\hline
				$624$& $1.0000745$ & $0.9999961$ & $4.63$\\
				$1,875$& $1.0000306$ & $1.0000028$ & $1.54$\\
				$3,125$& $1.0000066$ & $1.0000020$ & $0.92$\\
				$4,372$& $0.9999767$ & $0.9999901$ & $0.66$\\
				$5,628$& $1.0000056$ & $0.9999995$ & $0.51$\\
				$6,873$& $0.9999493$ & $1.0000012$ & $0.42$\\
				$8,122$& $1.0000025$ & $1.0000263$ & $0.36$\\
				$9,377$& $1.0000398$ & $0.9999828$ & $0.31$\\
				\hline
			\end{tabular}}
    \end{minipage}%
\hspace{4mm}
 \begin{minipage}{.47\textwidth}
 
		\resizebox{\textwidth}{!}{\begin{tabular}{|c| c c c|}
				\hline
				Sparsity& $\alpha$ & $\beta$ & $\epsilon_d$ \\
				\hline
				$680$& $1.0000611$ & $0.9999994$ & $5.18$\\
				$2,043$& $1.0000374$ & $0.9999985$ & $1.72$\\
				$3,404$& $1.0000435$ & $0.9999939$ & $1.03$\\
				$4,765$& $1.0000704$ & $1.0000090$ & $0.74$\\
				$6,126$& $1.0000432$ & $0.9999896$ & $0.57$\\
				$7,486$& $0.9999997$ & $1.0000215$ & $0.47$\\
				$8,848$& $1.0000726$ & $1.0000685$ & $0.40$\\
				$10,209$& $0.9999913$ & $1.0000321$ & $0.34$\\
				\hline
			\end{tabular}}
    \end{minipage}
	\caption{\small $\alpha$, $\beta$ (see \thmref{thm:spca:sdp:det}) and $\epsilon_d$ (the parameter of interest for~\cite{PDK2013}) for HGDP/HAPMAP chromosome~1 (left) and chromosome~2 (right) data.}
	\label{tab:constants}
\end{center}
\end{table*}


\section{Conclusion, limitations, and future work}
We present thresholding as a simple and intuitive approximation algorithm for SPCA, without imposing restrictive assumptions on the input covariance matrix. Our first algorithm provides runtime-vs-accuracy trade-offs and can be implemented in nearly input sparsity time; our second algorithm needs to solve an SDP and provides highly accurate solutions with novel theoretical guarantees. Our algorithms immediately extend to sparse kernel PCA. Our work does have limitations which are interesting topics for future work. First, is it possible to improve the accuracy guarantees of our SVD-based thresholding scheme to match its superior practical performance. Second, can we speed up our SDP-based thresholding scheme \textit{in practice} by using early termination of the SDP solvers or by using warm starts? Third, can we extend our approaches to handle more than one sparse singular vectors, by deflation or other strategies? Finally, it would be interesting to explore whether the proposed algorithms can approximately recover the \textit{support} of the vector $\xb^*$ (see eqn.~(\ref{eqn:spca})) instead of the optimal value $\mathcal{Z}^*$.

\section*{Broader Impacts and Limitations}
Our work is focused on speeding up and improving the accuracy of algorithms for SPCA. As such, it could have significant broader impacts by allowing users to more accurately solve SPCA problems like the ones discussed in our introduction. While applications of our work to real data could result in ethical considerations, this is an indirect (and unpredictable) side-effect of our work. Our experimental work uses publicly available datasets to evaluate the performance of our algorithms; no ethical considerations are raised.


\newpage

\bibliographystyle{plainnat}
\bibliography{PD-bib,references}

\clearpage
\appendix
\section{SPCA via thresholding: Proofs}\seclab{sxn:threshold:proofs}

We will use the notation of~\secref{sec:pcp}. 
For notational convenience, let $\sigma_1,\ldots,\sigma_n$ be the diagonal entries of the matrix $\Sigmab \in \mathbb{R}^{n \times n}$, i.e., the singular values of $\Ab$.

\lemszpcp*
\begin{proof}
Let $\Ub_{\ell,\perp} \in \mathbb{R}^{n \times (n-\ell)}$ be a matrix whose columns form a basis for the subspace perpendicular to the subspace spanned by the columns of $\Ub_{\ell}$. Similarly, let $\Sigmab_{\ell,\perp} \in \mathbb{R}^{(n-\ell) \times (n-\ell)}$ be the diagonal matrix of the bottom $n-\ell$ singular values of $\Ab$. Notice that $\Ub = [\Ub_{\ell}\ \ \Ub_{\ell,\perp}]$ and $\Sigmab = [\Sigmab_{\ell}\ \ \zero;\ \zero\ \ \Sigmab_{\ell,\perp}]$; thus,
\begin{flalign*}
\Ub\Sigmab^{1/2}\Ub^{\top} = \Ub_{\ell}\Sigmab_{\ell}^{1/2}\Ub_{\ell}^{\top} +
\Ub_{\ell,\perp}\Sigmab_{\ell,\perp}^{1/2}\Ub_{\ell,\perp}^{\top}.
\end{flalign*}
By the Pythagorean theorem,
\begin{align*}
\norm{\Ub\Sigmab^{1/2}\Ub^{\top}\x}_2^2=
\norm{\Ub_{\ell}\Sigmab_{\ell}^{1/2}\Ub_{\ell}^{\top}\x}_2^2 +
\norm{\Ub_{\ell,\perp}\Sigmab_{\ell,\perp}^{1/2}\Ub_{\ell,\perp}^{\top}\x}_2^2.
\end{align*}
Using invariance properties of the vector two-norm and sub-multiplicativity, we get
\begin{align*}
\norm{\Sigmab_{\ell}^{1/2}\Ub_{\ell}^{\top}\x}_2^2 \geq
\norm{\Sigmab^{1/2}\Ub^{\top}\x}_2^2-
\norm{\Sigmab_{\ell,\perp}^{1/2}}_2^2\norm{\Ub_{\ell,\perp}^{\top}\x}_2^2.
\end{align*}
We conclude the proof by noting that $\norm{\Sigmab^{1/2}\Ub^{\top}\x}_2^2 = \x^\top \Ub\Sigmab\Ub^{\top}\x = \x^\top \Ab \x$ and
$$\norm{\Sigmab_{\ell,\perp}^{1/2}}_2^2 = \sigma_{\ell+1} \leq \frac{1}{\ell}\sum_{i=1}^n \sigma_i = \frac{\trace(\Ab)}{\ell}.$$
The inequality above follows since $\sigma_1\geq \sigma_2\geq \ldots \sigma_\ell \geq \sigma_{\ell+1}\geq \ldots \geq \sigma_n$. We conclude the proof by setting $\ell = 1/\varepsilon$.
\end{proof}

\thmpcpmain*
\begin{proof}
Let $R = \{i_1,\ldots,i_{|R|}\}$ be the set of indices of rows of $\Ub_{\ell}$ (columns of $\Ub_{\ell}^\top$) that have squared norm at least $\varepsilon^2/k$ and let $\bar{R}$ be its complement. Here $|R|$ denotes the cardinality of the set $R$ and $R \cup \bar{R} = \{1,\ldots,n\}$. Let $\Rb \in \mathbb{R}^{n \times |R|}$ be the sampling matrix that selects the columns of $\Ub_{\ell}$ whose indices are in the set $R$ and let
$\Rb_{\perp} \in \mathbb{R}^{n \times (n-|R|)}$ be the sampling matrix that selects the columns of $\Ub_{\ell}$ whose indices are in the set $\bar{R}$.
Thus, each column of $\Rb$ (respectively $\Rb_{\perp}$) has a single non-zero entry, equal to one, corresponding to one of the $|R|$ (respectively $|\bar{R}|$) selected columns. Formally, $\Rb_{i_t,t}=1$ for all $t=1,\ldots, |R|$, while all other entries of $\Rb$ (respectively $\Rb_{\perp}$) are set to zero; $\Rb_{\perp}$ can be defined analogously.
The following properties are easy to prove: $\Rb\Rb^{\top}+\Rb_{\perp}\Rb_{\perp}^{\top} = \Ib_n$; $\Rb^{\top}\Rb=\Ib$; $\Rb_{\perp}^{\top}\Rb_{\perp}=\Ib$;
$\Rb_{\perp}^{\top}\Rb=\zero$.
Recall that $\xb^*$ is the optimal solution to the SPCA problem from eqn.~(\ref{eqn:spca}). We proceed as follows:

\begin{flalign}
\norm{\Sigmab_{\ell}^{1/2}\U_{\ell}^\top\x^*}_2^2&=\norm{\Sigmab_{\ell}^{1/2}\U_{\ell}^\top(\Rb\Rb^{\top}+\Rb_{\perp}\Rb_{\perp}^{\top})\x^*}_2^2\nonumber\\
&\leq \norm{\Sigmab_{\ell}^{1/2}\U_{\ell}^\top\Rb\Rb^{\top}\x^*}_2^2+\norm{\Sigmab_{\ell}^{1/2}\U_{\ell}^\top\Rb_{\perp}\Rb_{\perp}^{\top}\x^*}_2^2\nonumber\\
&+2\norm{\Sigmab_{\ell}^{1/2}\U_{\ell}^\top\Rb\Rb^{\top}\x^*}_2\norm{\Sigmab_{\ell}^{1/2}\U_{\ell}^\top\Rb_{\perp}\Rb_{\perp}^{\top}\x^*}_2\nonumber\\
&\leq \norm{\Sigmab_{\ell}^{1/2}\U_{\ell}^\top\Rb\Rb^{\top}\x^*}_2^2+\sigma_1\norm{\U_{\ell}^\top\Rb_{\perp}\Rb_{\perp}^{\top}\x^*}_2^2\nonumber\\
&+2\sigma_1\norm{\U_{\ell}^\top\Rb\Rb^{\top}\x^*}_2\norm{\U_{\ell}^\top\Rb_{\perp}\Rb_{\perp}^{\top}\x^*}_2\label{eqn:sdp:svd:1}.
\end{flalign}
The above inequalities follow from the Pythagorean theorem and sub-multiplicativity. 
We now bound the second term in the right-hand side of the above inequality.
\begin{flalign}
  \norm{\U_{\ell}^\top\Rb_{\perp}\Rb_{\perp}^{\top}\x^*}_2 &=  \|\sum_{i=1}^n (\U_{\ell}^\top\Rb_{\perp})_{*i}(\Rb_{\perp}^{\top}\x^*)_i\|_2\nonumber\\
  &\leq \sum_{i=1}^n \|(\U_{\ell}^\top\Rb_{\perp})_{*i}\|_2 \cdot |(\Rb_{\perp}^{\top}\x^*)_i|\leq \sqrt{\frac{\varepsilon^2}{k}}\sum_{i=1}^n |(\Rb_{\perp}^{\top}\x^*)_i|\nonumber\\
  &\leq \sqrt{\frac{\varepsilon^2}{k}}\|\Rb_{\perp}^{\top}\x^*\|_1\leq \sqrt{\frac{\varepsilon}{k}}\sqrt{k} = \varepsilon.\label{eqn:sdp:svd:2}
\end{flalign}
In the above derivations we use standard properties of norms and the fact that the columns of $\U_{\ell}^\top$ that have indices in the set $\bar{R}$ have squared norm at most $\varepsilon^2/k$. The last inequality follows from $\|\Rb_{\perp}^{\top}\x^*\|_1 \leq \|\x^*\|_1 \leq \sqrt{k}$, since $\xb^*$ has at most $k$ non-zero entries and Euclidean norm at most one.

Recall that the vector $\yb$ of~\algref{alg:sparse:pca} maximizes $\|\Sigmab_{\ell}^{1/2}\U_{\ell}^\top\Rb\xb\|_2$ over all vectors $\xb$ of appropriate dimensions (including $\Rb \x^*$) and thus
\begin{flalign}\label{eqn:sdp:svd:3}
  \|\Sigmab_{\ell}^{1/2}\U_{\ell}^\top\Rb\yb\|_2 \geq \norm{\Sigmab_{\ell}^{1/2}\U_{\ell}^\top\Rb\Rb^{\top}\x^*}_2.
\end{flalign}
Combining eqns.~(\ref{eqn:sdp:svd:1}),~(\ref{eqn:sdp:svd:2}), and~(\ref{eqn:sdp:svd:3}), we get that for sufficiently small $\eps$, 
\begin{flalign}\label{eqn:sdp:svd:4}
\norm{\Sigmab_{\ell}^{1/2}\U_{\ell}^\top\x^*}_2^2 \leq \|\Sigmab_{\ell}^{1/2}\U_{\ell}^\top\zb\|_2^2 + 2\varepsilon\trace(\Ab).
\end{flalign}
In the above we used $\zb = \Rb\yb$ (as in~\algref{alg:sparse:pca}) and $\sigma_1 \leq \trace(\Ab)$. Notice that
$$\U_{\ell}\Sigmab_{\ell}^{1/2}\U_{\ell}^\top\zb + \U_{\ell,\perp}\Sigmab_{\ell,\perp}^{1/2}\U_{\ell,\perp}^\top\zb=
\U\Sigmab^{1/2}\U^\top\zb,$$
and using the Pythagorean theorem we get
$$\|\U_{\ell}\Sigmab_{\ell}^{1/2}\U_{\ell}^\top\zb\|_2^2 + \|\U_{\ell,\perp}\Sigmab_{\ell,\perp}^{1/2}\U_{\ell,\perp}^\top\zb|_2^2=
\|\U\Sigmab^{1/2}\U^\top\zb\|_2^2.$$
Using the unitary invariance of the two norm and dropping a non-negative term, we get the bound
\begin{flalign}\label{eqn:sdp:svd:5}
\|\Sigmab_{\ell}^{1/2}\U_{\ell}^\top\zb\|_2^2 \leq
\|\Sigmab^{1/2}\U^\top\zb\|_2^2.
\end{flalign}
Combining eqns.~(\ref{eqn:sdp:svd:4}) and~(\ref{eqn:sdp:svd:5}), we conclude
\begin{flalign}\label{eqn:sdp:svd:6}
\norm{\Sigmab_{\ell}^{1/2}\U_{\ell}^\top\x^*}_2^2 \leq \|\Sigmab^{1/2}\U^\top\zb\|_2^2 + 2\varepsilon\trace(\Ab).
\end{flalign}
We now apply~\lemref{lem:sz:pcp} to the optimal vector $\xb^*$ to get
$$\norm{\Sigmab^{1/2} \Ub^\top \xb^*}_2^2- \eps\trace(\A) \leq \norm{\Sigmab_{\ell}^{1/2} \Ub_{\ell}^\top\xb^*}_2^2.$$
Combining with eqn.~(\ref{eqn:sdp:svd:6}) we get
$$\zb^\top \Ab \zb \geq\mathcal{Z}^* - 3\eps\trace(\A).$$
In the above we used
$\|\Sigmab^{1/2}\U^\top\zb\|_2^2 = \zb^\top \Ab \zb$ and
$\norm{\Sigmab^{1/2} \Ub^\top\xb^*}_2^2=(\xb^*)^\top \Ab \xb^* = \mathcal{Z}^*$. 
The result then follows from rescaling $\eps$. 
\end{proof}

\section{Additional Notes on Experiments}\label{app:expt}
For~\algref{alg:sparse:pca}, we fix the threshold parameter $\ell$ to $1$ for all datasets. For \algref{alg:sparse:sdp:det}, we rely on ADMM-based first-order methods to solve eqn.~\eqref{eqn:spdrelax}. More precisely, we use the \texttt{admm.spca()} function of the \texttt{ADMM} package in R~\citep{ma2013alternating} as well as Python's \texttt{cvxpy} package with \texttt{SCS} solver to solve eqn.~\eqref{eqn:spdrelax}. 
%
Next, we show the performance of our algorithms on \emph{pitprops} data:

\begin{table*}[ht]
	\begin{center}
		\resizebox{\textwidth}{!}{\begin{tabular}{|l| c c c c c c c c c c c c c |c c|}
				\hline
				~~& topdiam & length & moist & testsg & ovensg & ringtop & ringbut & bowmax & bowdist & whorls & clear& knots  & diaknot & PVE & $\mathcal{Z}^*$\\
				\hline
				\texttt{spca-svd}~~~~~~(\algref{alg:sparse:pca})& $~~0.420$ & $~~0.422$ & $0$ & $0$ & $0$ & $~~0.296$ & $~~0.416$ & $~~0.305$ & $~~0.371$ & $~~0.394$ &$0$ &$0$ &$0$ & $30.71\%$ & $3.993$\\
				\texttt{spca-sdp}~(\algref{alg:sparse:sdp:det})& $~~0.424$ & $~~0.430$ & $0$ & $0$ & $0$ & $~~0.268$ & $~~0.403$ & $~~0.313$ & $~~0.379$ & $~~0.399$ & $0$&$0$ &$0$ & $30.74\%$ & $3.996$\\
				\texttt{dec}~\citep{Yuan_2019_CVPR} & $-0.423$ & $-0.430$ & $0$ & $0$ & $0$ & $-0.268$ & $-0.403$ & $-0.313$ & $-0.379$ & $-0.399$ & 0 & 0 & 0 & $30.74\%$ & $3.996$\\
				\texttt{cwpca}~\citep{beck2016sparse} & $-0.423$ & $-0.430$ & $0$ & $0$ & $0$ & $-0.268$ & $-0.403$ & $-0.313$ & $-0.379$ & $-0.399$ & 0 & 0 & 0 & $30.74\%$ & $3.996$\\
				\texttt{spca-lowrank}~\citep{PDK2013} & $-0.427$ & $-0.432$ & $0$ & $0$ & $0$ & $-0.249$ & $-0.390$ & $-0.326$ & $-0.383$ & $-0.403$ & 0 & 0 & 0 & $30.72\%$ & $3.994$\\
				\hline
			\end{tabular}
		}
	\end{center}
	\caption{\small Loadings, \% of variance explained (PVE), and the objective function value for the first principal component of the Pit Props data.}
	\label{tab:multicol1}
\end{table*}

\subsection{Real Data}\label{app:real}
\textbf{Population genetics data.} We use population genetics data from the Human Genome Diversity Panel~\citep{int07} and the HAPMAP~\citep{li2008}. In particular, we use the 22 matrices (one for each chromosome) that encode all autosomal genotypes. Each matrix contains 2,240 rows and a varying number of columns that is equal to the number of single nucleotide polymorphisms (SNPs, well-known biallelic loci of genetic variation across the human genome) in the respective chromosome. The columns of each matrix were mean-centered as a preprocessing step. See Table~\ref{tab:SNPs} for summary statistics.

\textbf{Gene expression data.}
We also use a lung cancer gene expression dataset (GSE10072) from the NCBI Gene Expression Omnibus database~\citep{landi2008gene}. This dataset contains $107$ samples (58 cases and 49 controls) and 22,215 features. Both the population genetics and the gene expression datasets are interesting in the context of sparse PCA beyond numerical evaluations, since the sparse components can be directly interpreted to identify small sets of SNPs or genes that capture the data variance.


\subsection{Synthetic Data} \label{app:syn}
We also use a synthetic dataset generated  using the same mechanism as in~\citep{FountoulakisKKD17}. Specifically, we construct the $m\times n$ matrix $\Xb$ such that $\Xb=\Ub\Sigmab\Vb^\top+\Eb_\sigma$. Here, $\Eb_\sigma$ is a noise matrix, containing i.i.d. Gaussian elements with zero mean and we set $\sigma=10^{-3}$; $\Ub\in\mathbb{R}^{m\times m}$ is a Hadamard matrix with normalized columns; $\Sigmab=(\tilde{\Sigmab}~~\zero)\in\mathbb{R}^{m\times n}$ such that $\tilde{\Sigmab}\in\mathbb{R}^{m\times m}$ is a  diagonal matrix with $\tilde{\Sigmab}_{11}=100$ and $\tilde{\Sigmab}_{ii}=e^{-i}$ for $i=2,\dots ,m$; $\Vb\in\mathbb{R}^{n\times n}$ such that $\Vb=\Gb_n(\theta)\tilde{\Vb}$, where $\tilde{\Vb}\in\mathbb{R}^{n\times n}$ is also a Hadamard matrix with normalized columns and 
$$\Gb_n(\theta)=\Gb (i_1,i_1+1,\theta)\,\Gb(i_2,i_2+1,\theta)\dots \Gb(i_{n/4},i_{n/4}+1,\theta),$$
is a composition of $\frac{n}{4}$ Givens rotation matrices with $i_k=\frac{n}{2}+2k-1$ for $k=1,2,\dots, \frac{n}{4}$. Here $\Gb(i,j,\theta)\in\mathbb{R}^{n\times n}$ is a Givens rotation matrix, which rotates the plane
$i-j$ by an angle $\theta$. For $\theta\approx 0.27\pi$ and $n=2^{12}$, the matrix $\Gb_n(\theta)$ rotates the bottom $\frac{n}{2}$ components of the columns of $\tilde{\Vb}$, making half of them almost zero and the remaining half larger. Figure~\ref{fig:fig4} shows the absolute values of the elements of the first column of the matrices $\Vb$ and $\tilde{\Vb}$.

\subsection{Additional Experiments}\label{app:exp_addl}

In our additional experiments on the large datasets, Figure~\ref{fig:sfig7} shows the performance of various SPCA algorithms on synthetic data. We observe our algorithms perform optimally and  closely match with \texttt{dec}, \texttt{cwpca}, and \texttt{spca-lowrank}.
Moreover, notice that turning the bottom $\frac{n}{4}$ elements of $\tilde{\Vb}$ into large values doesn't affect the performances of \texttt{spca-svd} and \texttt{spca-sdp}, which further highlights the robustness of our methods. 
In Figure~\ref{fig:fig3}, we demonstrate how our algorithms perform on \textsc{Chr}~3 and \textsc{Chr}~4 of the population genetics data. We see a similar behavior as observed for \textsc{Chr}~1 and \textsc{Chr}~2 in Figures~\ref{fig:sfig1}-\ref{fig:sfig2}. 

\begin{figure}[!htb]
\begin{subfigure}{.5\textwidth}
  \centering
  \includegraphics[width=1\linewidth]{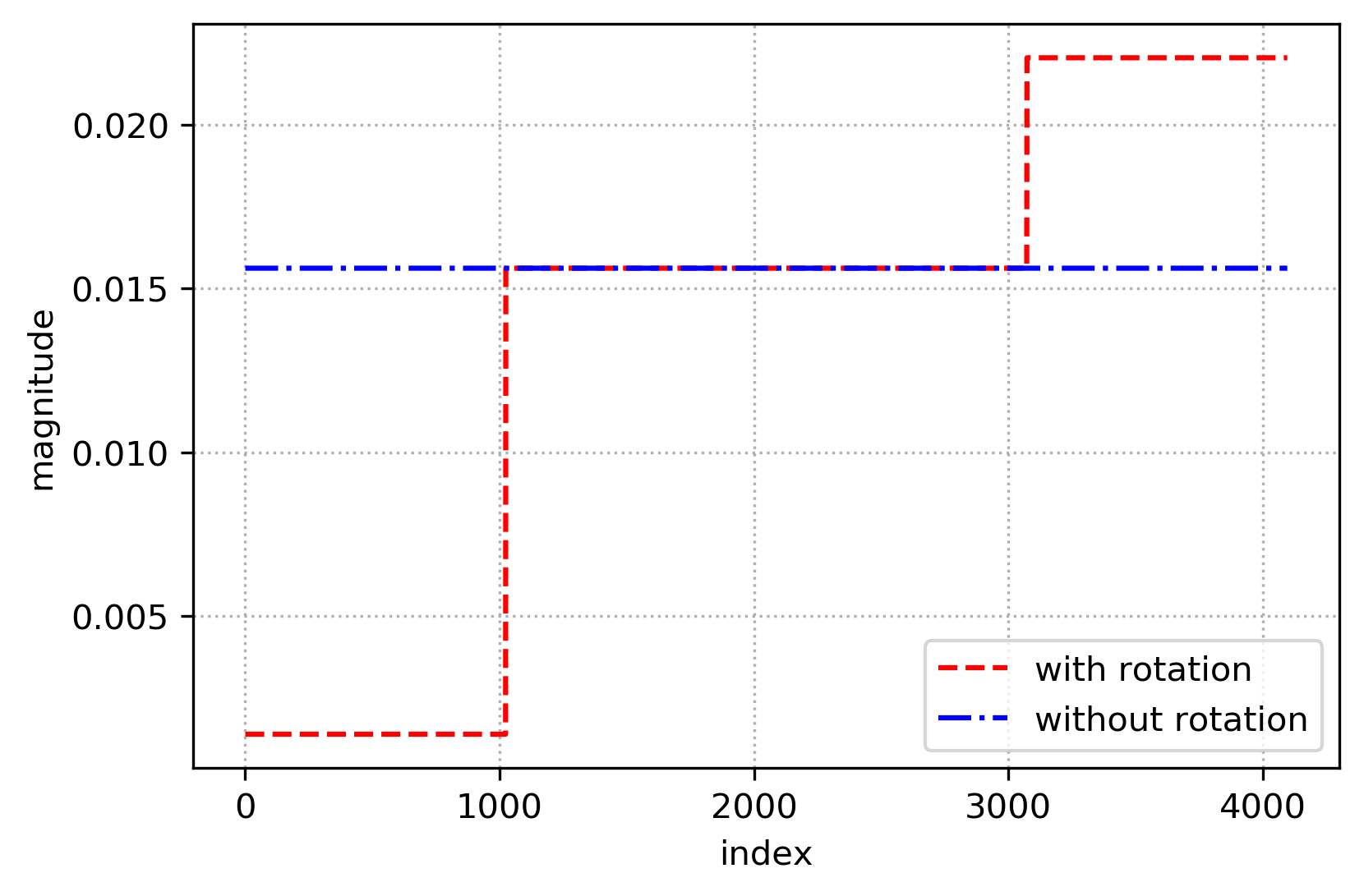}
  \caption{}
  \label{fig:sfig6}
\end{subfigure}%
\begin{subfigure}{.5\textwidth}
  \centering
  \includegraphics[width=1\linewidth]{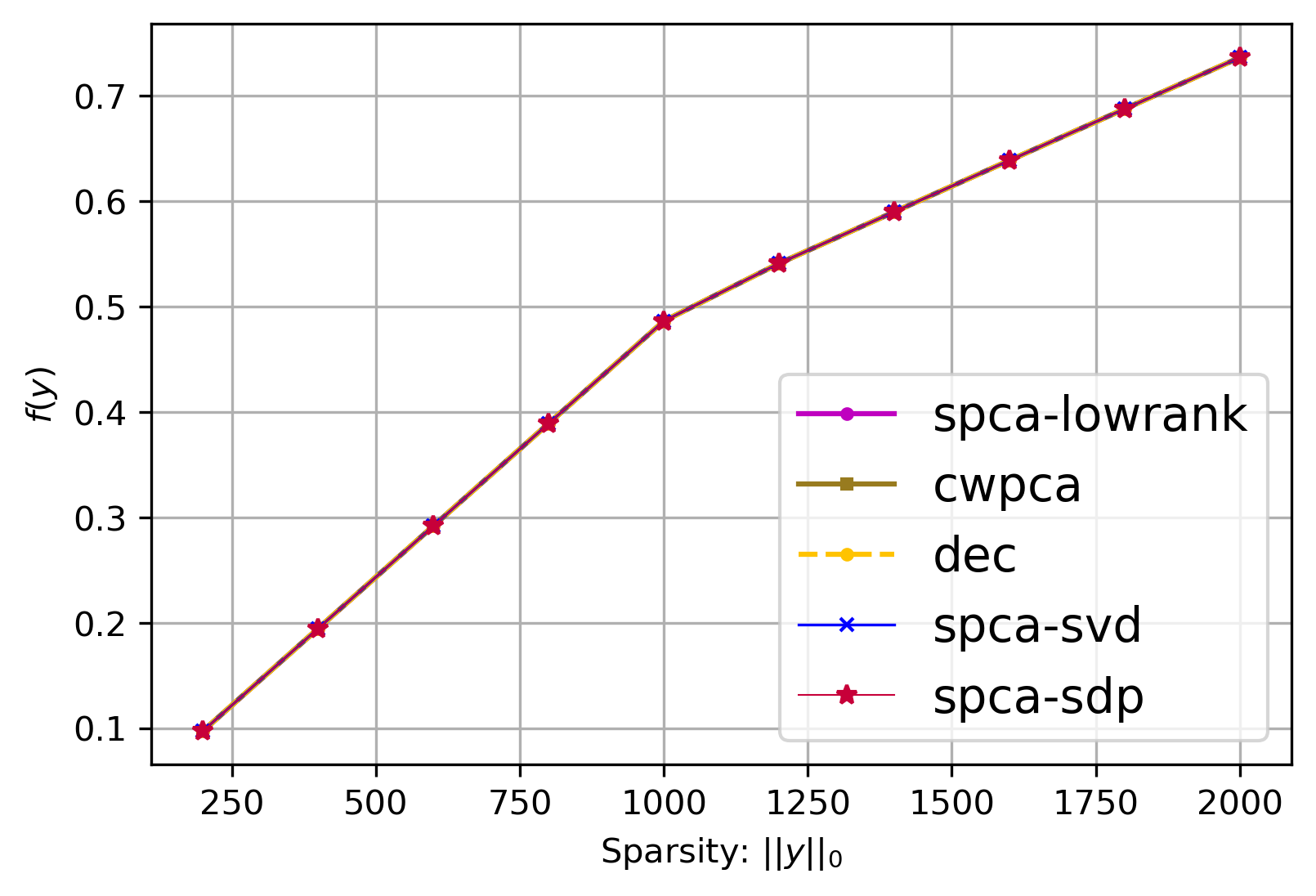}
  \caption{}
  \label{fig:sfig7}
\end{subfigure}
\caption{Experimental results on synthetic data with $m=2^7$ and $n=2^{12}$: (a) the red and the blue lines are the sorted absolute values of the elements of the first column
of matrices $\Vb$ and $\tilde{\Vb}$ respectively. (b) $f(\y)$ vs. sparsity ratio.}
\label{fig:fig4}
\end{figure}
\begin{figure}[!htb]
\begin{subfigure}{.32\textwidth}
  \centering
  \includegraphics[width=1\linewidth]{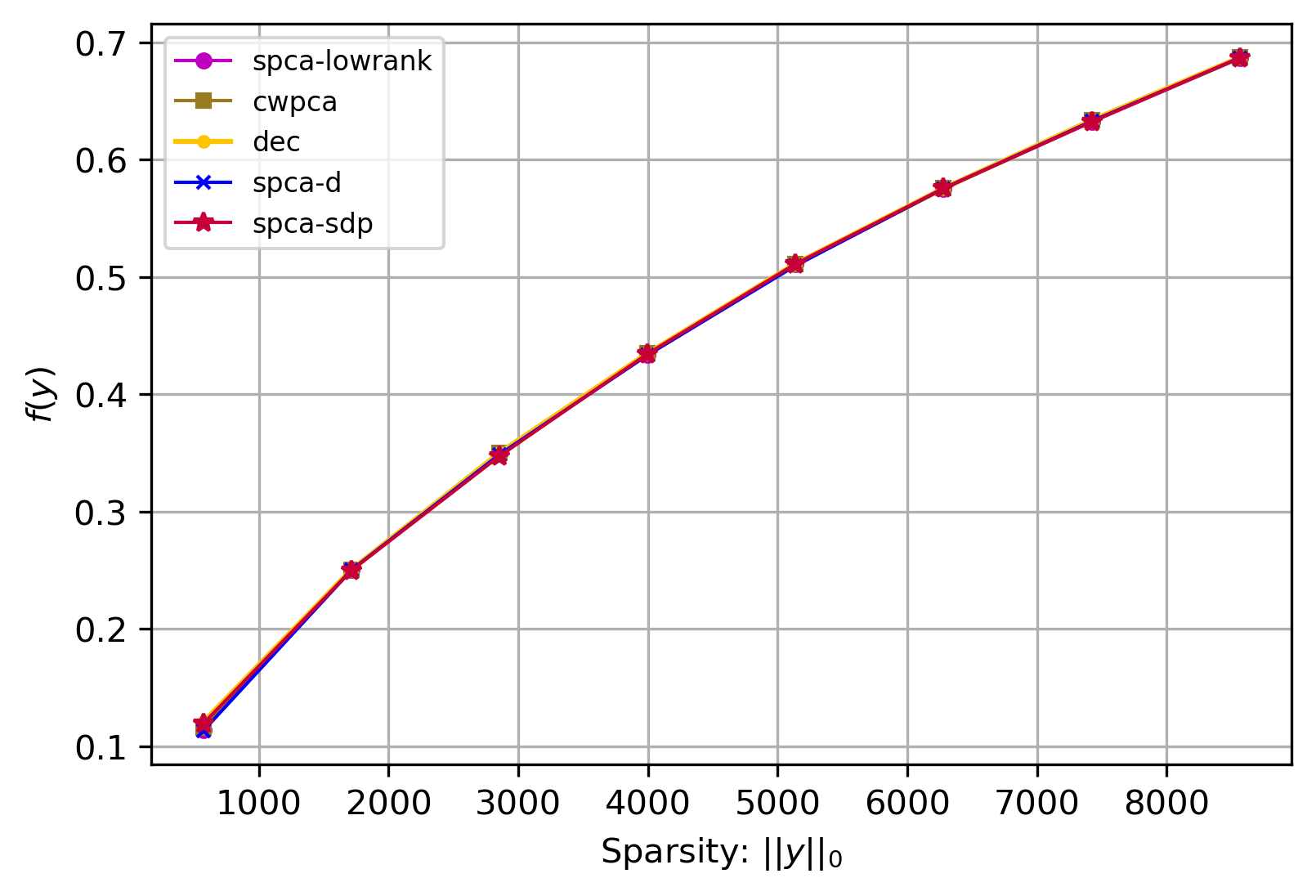}
  \caption{\textsc{Chr}~3, $n=34,258$}
  \label{fig:sfig4}
\end{subfigure}%
\begin{subfigure}{.32\textwidth}
  \centering
  \includegraphics[width=1\linewidth]{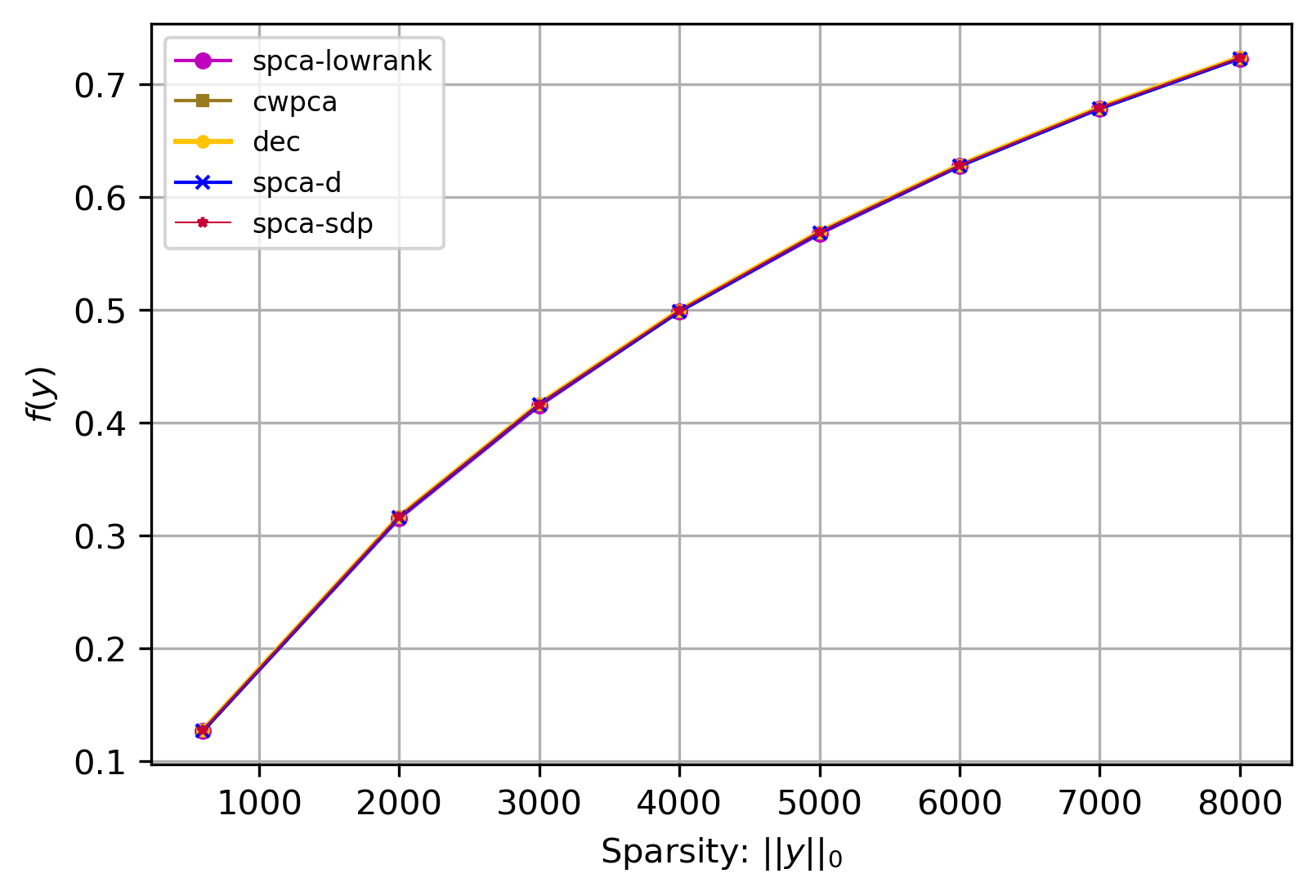}
  \caption{\textsc{Chr}~4, $n=30,328$}
  \label{fig:sfig5}
\end{subfigure}
\begin{subfigure}{.32\textwidth}
  \centering
  \includegraphics[width=1\linewidth]{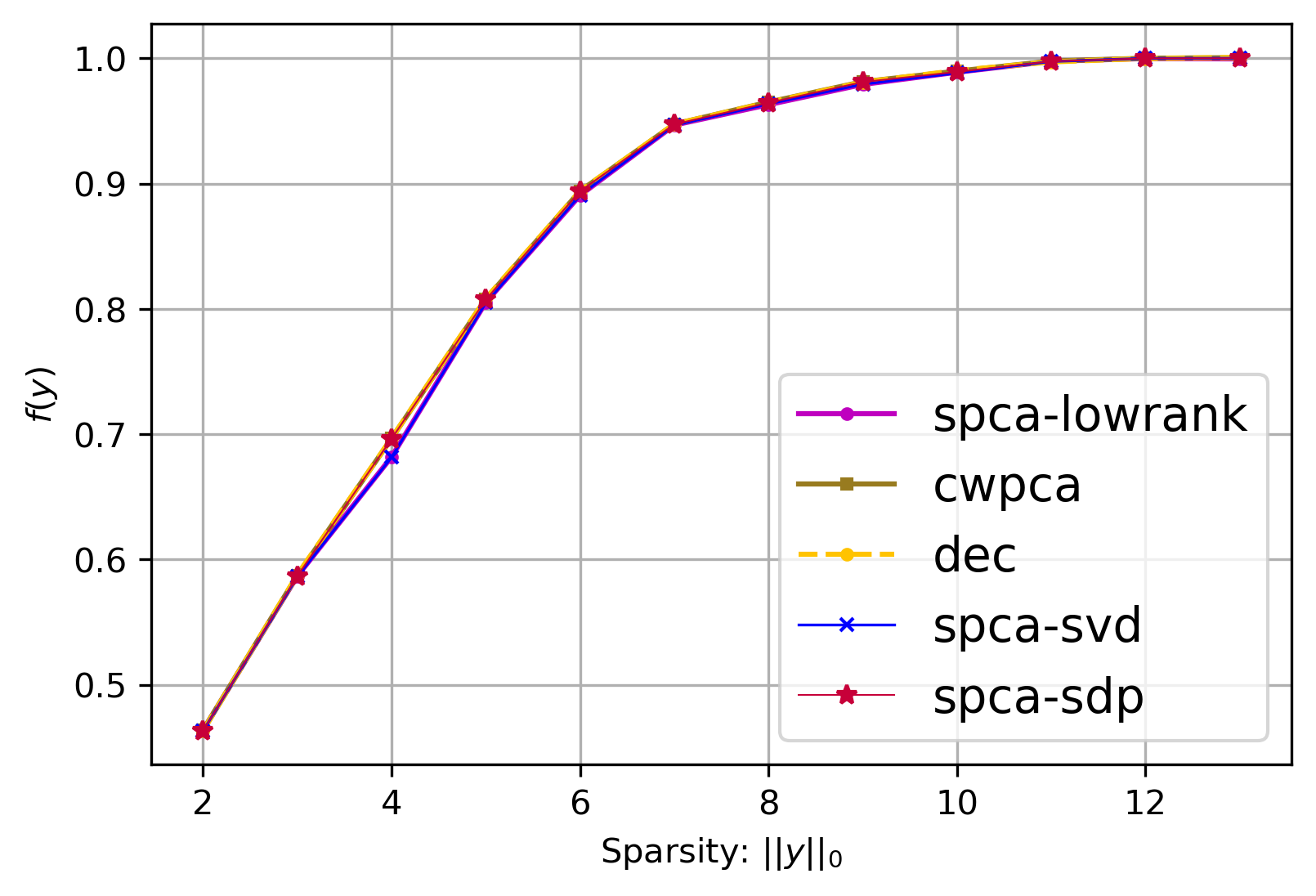}
  \caption{Pitprops data}
  \label{fig:sfig6}
\end{subfigure}
\caption{Experimental results on real data: $f(\y)$ vs. sparsity ratio.}
\label{fig:fig3}
\end{figure}

\begin{table}[!htb]
\centering
\caption{Statistics of the population genetics data.}
\label{tab:SNPs}
\begin{tabular}{rrcc}
\toprule
Dataset & \#\,Rows & \#\,Columns & Density\\
\midrule
\textsc{Chr}~1 & 2,240 & 37,493 & 0.986 \\
\textsc{Chr}~2 & 2,240 & 40,844 & 0.987 \\
\textsc{Chr}~3 & 2,240 & 34,258 & 0.986 \\
\textsc{Chr}~4 & 2,240 & 30,328 & 0.986 \\
Gene expression & 107& 22,215 & 0.999 \\
\bottomrule
\end{tabular}
\end{table}

\end{document}

%% file: preamble.tex
\usepackage{amsmath,amssymb,amsfonts}
\usepackage{latexsym}
\usepackage{subcaption}
\usepackage{graphicx}
\usepackage{booktabs}
\usepackage[backref=page]{hyperref}
\usepackage{url}
\usepackage{color}
\usepackage{wrapfig}
\usepackage{tikz}
\usetikzlibrary{decorations.pathreplacing}
\usepackage{setspace}
\usepackage{algorithm}
\usepackage[noend]{algpseudocode}
\usepackage[framemethod=tikz]{mdframed}
\usepackage{xspace}
\usepackage{pgfplots}
\usepackage{framed}
\usepackage{thmtools,thm-restate}
\usepackage{nicefrac}
\pgfplotsset{compat=1.5}
\allowdisplaybreaks

\newenvironment{proof}{\noindent{\bf Proof : \ }}{\hfill$\Box$\par\medskip}

\newtheorem{theorem}{Theorem}[section]

\newenvironment{proofof}[1]{\begin{trivlist} \item {\bf Proof
#1:~~}}
  {\qed\end{trivlist}}

\newcommand{\namedref}[2]{\hyperref[#2]{#1~\ref*{#2}}}
\newcommand{\thmlab}[1]{\label{thm:#1}}
\newcommand{\thmref}[1]{\namedref{Theorem}{thm:#1}}
\newcommand{\lemlab}[1]{\label{lem:#1}}
\newcommand{\lemref}[1]{\namedref{Lemma}{lem:#1}}

\newcommand{\seclab}[1]{\label{sec:#1}}
\newcommand{\secref}[1]{\namedref{Section}{sec:#1}}

\newcommand{\alglab}[1]{\label{alg:#1}}
\renewcommand{\algref}[1]{\namedref{Algorithm}{alg:#1}}

\newcommand\norm[1]{\left\lVert#1\right\rVert}

\renewcommand{\O}[1]{\ensuremath{\mathcal{O}\left(#1\right)}}

\newcommand{\eps}{\varepsilon}
\newcommand{\nnz}{\text{nnz}}

\def \Sigmab    {\mdef{\mathbf{\Sigma}}}

\def \A    {\mdef{\mathbf{A}}}

\def \U    {\mdef{\mathbf{U}}}
\def \X    {\mdef{\mathbf{X}}}

\def \Z    {\mdef{\mathbf{Z}}}

\def \x    {\mdef{\mathbf{x}}}
\def \y    {\mdef{\mathbf{y}}}
\def \z    {\mdef{\mathbf{z}}}

\newcommand{\zero}{\mathbf{0}}

\newcommand{\mdef}[1]{{\ensuremath{#1}}\xspace}  

\DeclareMathOperator*{\argmax}{argmax}

\DeclareMathOperator*{\poly}{poly}
\DeclareMathOperator*{\trace}{Tr}

\newcommand{\abs}[1]{\mdef{\left|#1\right|}}         

\newcommand{\ignore}[1]{}

\newif\ifnotes\notestrue 
\ifnotes
\newcommand{\samson}[1]{\textcolor{purple}{{\bf (Samson:} {#1}{\bf ) }} \marginpar{\tiny\bf
             \begin{minipage}[t]{0.5in}
               \raggedright S:
            \end{minipage}}}            							
\else
\newcommand{\samson}[1]{}
\fi

\definecolor{darkred}{rgb}{0.55, 0.0, 0.0}

\hypersetup{
     colorlinks   = true,
     citecolor    = mahogany,
     linkcolor	  = darkblue
}

\makeatletter
\renewcommand*{\@fnsymbol}[1]{\textcolor{darkred}{\ensuremath{\ifcase#1\or *\or \dagger\or \ddagger\or
 \mathsection\or \triangledown\or \mathparagraph\or \|\or **\or \dagger\dagger
   \or \ddagger\ddagger \else\@ctrerr\fi}}}
\makeatother

\providecommand{\email}[1]{\href{mailto:#1}{\nolinkurl{#1}\xspace}}

\definecolor{mahogany}{rgb}{0.75, 0.25, 0.0}
\definecolor{darkblue}{rgb}{0.0, 0.0, 0.55}

\newcommand{\ub}{\mathbf{u}}

\newcommand{\xb}{\mathbf{x}}
\newcommand{\yb}{\mathbf{y}}
\newcommand{\zb}{\mathbf{z}}

\newcommand{\Ab}{\mathbf{A}}

\newcommand{\Eb}{\mathbf{E}}

\newcommand{\Gb}{\mathbf{G}}

\newcommand{\Ib}{\mathbf{I}}

\newcommand{\Rb}{\mathbf{R}}

\newcommand{\Ub}{\mathbf{U}}
\newcommand{\Vb}{\mathbf{V}}

\newcommand{\Xb}{\mathbf{X}}

\newcommand{\Zb}{\mathbf{Z}}

%% file: spca-main.bbl
\begin{thebibliography}{34}
\providecommand{\natexlab}[1]{#1}
\providecommand{\url}[1]{\texttt{#1}}
\expandafter\ifx\csname urlstyle\endcsname\relax
  \providecommand{\doi}[1]{doi: #1}\else
  \providecommand{\doi}{doi: \begingroup \urlstyle{rm}\Url}\fi

\bibitem[Alizadeh(1995)]{Ali1995}
Farid Alizadeh.
\newblock {Interior Point Methods in Semidefinite Programming with Applications
  to Combinatorial Optimization}.
\newblock \emph{SIAM Journal on Optimization}, 5\penalty0 (1):\penalty0 13--51,
  1995.

\bibitem[Amini and Wainwright(2009)]{Amini09}
Arash~A. Amini and Martin~J. Wainwright.
\newblock {High-dimensional Analysis of Semidefinite Relaxations for Sparse
  Principal Components}.
\newblock \emph{Annals of Statistics}, 37:\penalty0 2877--2921, 2009.

\bibitem[Ando et~al.(2009)Ando, Nakata, and Yamashita]{ando2009approximating}
Ei~Ando, Toshio Nakata, and Masafumi Yamashita.
\newblock {Approximating the Longest Path Length of a Stochastic DAG by a
  Normal Distribution in Linear Time}.
\newblock \emph{Journal of Discrete Algorithms}, 7\penalty0 (4):\penalty0
  420--438, 2009.

\bibitem[Asteris et~al.(2011)Asteris, Papailiopoulos, and
  Karystinos]{asteris2011sparse}
Megasthenis Asteris, Dimitris Papailiopoulos, and George~N Karystinos.
\newblock {Sparse Principal Component of a Rank-deficient Matrix}.
\newblock In \emph{2011 IEEE International Symposium on Information Theory
  Proceedings}, pages 673--677, 2011.

\bibitem[Asteris et~al.(2015)Asteris, Papailiopoulos, Kyrillidis, and
  Dimakis]{asteris2015sparse}
Megasthenis Asteris, Dimitris Papailiopoulos, Anastasios Kyrillidis, and
  Alexandros~G Dimakis.
\newblock {Sparse {PCA} via Bipartite Matchings}.
\newblock In \emph{Advances in Neural Information Processing Systems}, pages
  766--774, 2015.

\bibitem[Beck and Vaisbourd(2016)]{beck2016sparse}
Amir Beck and Yakov Vaisbourd.
\newblock The {S}parse {P}rincipal {C}omponent {A}nalysis {P}roblem:
  {O}ptimality {C}onditions and {A}lgorithms.
\newblock \emph{Journal of Optimization Theory and Applications}, 170\penalty0
  (1):\penalty0 119--143, 2016.

\bibitem[Cadima and Jolliffe(1995)]{CJ95}
Jorge Cadima and Ian~T. Jolliffe.
\newblock {Loading and Correlations in the Interpretation of Principal
  Components}.
\newblock \emph{Journal of Applied Statistics}, 22\penalty0 (2):\penalty0
  203--214, 1995.

\bibitem[Chan et~al.(2016)Chan, Papailliopoulos, and Rubinstein]{ChanPR16}
Siu~On Chan, Dimitris Papailliopoulos, and Aviad Rubinstein.
\newblock {On the Approximability of Sparse {PCA}}.
\newblock In \emph{Proceedings of the 29th Conference on Learning Theory},
  pages 623--646, 2016.

\bibitem[Consortium(2007)]{int07}
International~HapMap Consortium.
\newblock {A Second Generation Human Haplotype Map of over 3.1 Million {SNP}s}.
\newblock \emph{Nature}, 449\penalty0 (7164):\penalty0 851, 2007.

\bibitem[d'Aspremont et~al.(2007)d'Aspremont, Ghaoui, Jordan, and
  Lanckriet]{AGJ2007}
Alexandre d'Aspremont, Laurent~El Ghaoui, Michael~I. Jordan, and Gert R.~G.
  Lanckriet.
\newblock {A Direct Formulation for Sparse {PCA} using Semidefinite
  Programming}.
\newblock \emph{{SIAM} Review}, 49\penalty0 (3):\penalty0 434--448, 2007.

\bibitem[d’Aspremont et~al.(2008)d’Aspremont, Bach, and
  Ghaoui]{d2008optimal}
Alexandre d’Aspremont, Francis Bach, and Laurent~El Ghaoui.
\newblock {Optimal Solutions for Sparse Principal Component Analysis}.
\newblock \emph{Journal of Machine Learning Research}, 9\penalty0
  (Jul):\penalty0 1269--1294, 2008.

\bibitem[d’Aspremont et~al.(2014)d’Aspremont, Bach, and
  El~Ghaoui]{d2014approximation}
Alexandre d’Aspremont, Francis Bach, and Laurent El~Ghaoui.
\newblock {Approximation Bounds for Sparse Principal Component Analysis}.
\newblock \emph{Mathematical Programming}, 148\penalty0 (1-2):\penalty0
  89--110, 2014.

\bibitem[d'Orsi et~al.(2020)d'Orsi, Kothari, Novikov, and Steurer]{dOrsiKNS20}
Tommaso d'Orsi, Pravesh~K. Kothari, Gleb Novikov, and David Steurer.
\newblock Sparse {PCA:} algorithms, adversarial perturbations and certificates.
\newblock In \emph{61st {IEEE} Annual Symposium on Foundations of Computer
  Science, {FOCS}}, pages 553--564, 2020.

\bibitem[Fountoulakis et~al.(2017)Fountoulakis, Kundu, Kontopoulou, and
  Drineas]{FountoulakisKKD17}
Kimon Fountoulakis, Abhisek Kundu, Eugenia{-}Maria Kontopoulou, and Petros
  Drineas.
\newblock {A Randomized Rounding Algorithm for Sparse {PCA}}.
\newblock \emph{ACM Transactions on Knowledge Discovery from Data}, 11\penalty0
  (3):\penalty0 38:1--38:26, 2017.

\bibitem[Jeffers(1967)]{jeffers1967two}
John~NR Jeffers.
\newblock Two case studies in the application of principal component analysis.
\newblock \emph{Journal of the Royal Statistical Society: Series C (Applied
  Statistics)}, 16\penalty0 (3):\penalty0 225--236, 1967.

\bibitem[Jolliffe(1995)]{jolliffe1995rotation}
Ian~T. Jolliffe.
\newblock {Rotation of principal components: Choice of Normalization
  Constraints}.
\newblock \emph{Journal of Applied Statistics}, 22\penalty0 (1):\penalty0
  29--35, 1995.

\bibitem[Jolliffe et~al.(2003)Jolliffe, Trendafilov, and
  Uddin]{jolliffe2003modified}
Ian~T. Jolliffe, Nickolay~T. Trendafilov, and Mudassir Uddin.
\newblock {A Modified Principal Component Technique Based on the {LASSO}}.
\newblock \emph{Journal of Computational and Graphical Statistics}, 12\penalty0
  (3):\penalty0 531--547, 2003.

\bibitem[Journ{\'e}e et~al.(2010)Journ{\'e}e, Nesterov, Richt{\'a}rik, and
  Sepulchre]{journee2010generalized}
Michel Journ{\'e}e, Yurii Nesterov, Peter Richt{\'a}rik, and Rodolphe
  Sepulchre.
\newblock {Generalized Power Method for Sparse Principal Component Analysis}.
\newblock \emph{Journal of Machine Learning Research}, 11\penalty0 (2), 2010.

\bibitem[Kuleshov(2013)]{Kuleshov13}
Volodymyr Kuleshov.
\newblock {Fast Algorithms for Sparse Principal Component Analysis Based on
  {R}ayleigh Quotient Iteration}.
\newblock In \emph{Proceedings of the 30th International Conference on Machine
  Learning}, pages 1418--1425, 2013.

\bibitem[Landi et~al.(2008)Landi, Dracheva, Rotunno, Figueroa, Liu, Dasgupta,
  Mann, Fukuoka, Hames, Bergen, et~al.]{landi2008gene}
Maria~Teresa Landi, Tatiana Dracheva, Melissa Rotunno, Jonine~D. Figueroa,
  Huaitian Liu, Abhijit Dasgupta, Felecia~E. Mann, Junya Fukuoka, Megan Hames,
  Andrew~W. Bergen, et~al.
\newblock {Gene Expression Signature of Cigarette Smoking and Its Role in Lung
  Adenocarcinoma Development and Survival}.
\newblock \emph{PloS one}, 3\penalty0 (2), 2008.

\bibitem[Li et~al.(2008)Li, Absher, Tang, Southwick, Casto, Ramachandran, Cann,
  Barsh, Feldman, Cavalli-Sforza, et~al.]{li2008}
Jun~Z. Li, Devin~M. Absher, Hua Tang, Audrey~M. Southwick, Amanda~M. Casto,
  Sohini Ramachandran, Howard~M. Cann, Gregory~S. Barsh, Marcus Feldman,
  Luigi~L. Cavalli-Sforza, et~al.
\newblock {Worldwide Human Relationships Inferred from Genome-Wide Patterns of
  Variation}.
\newblock \emph{Science}, 319\penalty0 (5866):\penalty0 1100--1104, 2008.

\bibitem[Ma(2013)]{ma2013alternating}
Shiqian Ma.
\newblock Alternating {D}irection {M}ethod of {M}ultipliers for {S}parse
  {P}rincipal {C}omponent {A}nalysis.
\newblock \emph{Journal of the Operations Research Society of China},
  1\penalty0 (2):\penalty0 253--274, 2013.

\bibitem[Magdon-Ismail(2017)]{magdon2017np}
Malik Magdon-Ismail.
\newblock {{NP}-Hardness and Inapproximability of Sparse {PCA}}.
\newblock \emph{Information Processing Letters}, 126:\penalty0 35--38, 2017.

\bibitem[Mahoney and Drineas(2009)]{MD09}
Michael~W. Mahoney and P.~Drineas.
\newblock {{CUR} Matrix Decompositions for Improved Data Analysis}.
\newblock In \emph{Proceedings of the National Academy of Sciences}, pages
  697--702, 106 (3), 2009.

\bibitem[Moghaddam et~al.(2006{\natexlab{a}})Moghaddam, Weiss, and
  Avidan]{MWA2006}
Baback Moghaddam, Yair Weiss, and Shai Avidan.
\newblock {Generalized Spectral Bounds for Sparse {LDA}}.
\newblock In \emph{Proceedings of the 23rd International Conference on Machine
  learning}, pages 641--648, 2006{\natexlab{a}}.

\bibitem[Moghaddam et~al.(2006{\natexlab{b}})Moghaddam, Weiss, and
  Avidan]{moghaddam2006spectral}
Baback Moghaddam, Yair Weiss, and Shai Avidan.
\newblock {Spectral Bounds for Sparse {PCA}: Exact and Greedy Algorithms}.
\newblock In \emph{Advances in Neural Information Processing Systems}, pages
  915--922, 2006{\natexlab{b}}.

\bibitem[Musco and Musco(2015)]{MuscoM15}
Cameron Musco and Christopher Musco.
\newblock Randomized block krylov methods for stronger and faster approximate
  singular value decomposition.
\newblock In \emph{Advances in Neural Information Processing Systems 28: Annual
  Conference on Neural Information Processing Systems}, pages 1396--1404, 2015.

\bibitem[Papailiopoulos et~al.(2013)Papailiopoulos, Dimakis, and
  Korokythakis]{PDK2013}
Dimitris Papailiopoulos, Alexandros Dimakis, and Stavros Korokythakis.
\newblock {Sparse {PCA} through Low-rank Approximations}.
\newblock In \emph{Proceedings of the 30th International Conference on Machine
  Learning}, pages 747--755, 2013.

\bibitem[Shen and Huang(2008)]{shen2008sparse}
Haipeng Shen and Jianhua~Z. Huang.
\newblock {Sparse Principal Component Analysis via Regularized Low Rank Matrix
  Approximation}.
\newblock \emph{Journal of Multivariate Analysis}, 99\penalty0 (6):\penalty0
  1015--1034, 2008.

\bibitem[Sriperumbudur et~al.(2007)Sriperumbudur, Torres, and
  Lanckriet]{sriperumbudur2007sparse}
Bharath~K. Sriperumbudur, David~A. Torres, and Gert~R.G. Lanckriet.
\newblock {Sparse Eigen Methods by {D.C.} Programming}.
\newblock In \emph{Proceedings of the 24th International Conference on Machine
  Learning}, pages 831--838, 2007.

\bibitem[Yuan et~al.(2019)Yuan, Shen, and Zheng]{Yuan_2019_CVPR}
Ganzhao Yuan, Li~Shen, and Wei-Shi Zheng.
\newblock A {D}ecomposition {A}lgorithm for the {S}parse {G}eneralized
  {E}igenvalue {P}roblem.
\newblock In \emph{Proceedings of the IEEE/CVF Conference on Computer Vision
  and Pattern Recognition (CVPR)}, 2019.

\bibitem[Yuan and Zhang(2013)]{yuan2013truncated}
Xiao-Tong Yuan and Tong Zhang.
\newblock {Truncated Power Method for Sparse Eigenvalue Problems}.
\newblock \emph{Journal of Machine Learning Research}, 14\penalty0
  (Apr):\penalty0 899--925, 2013.

\bibitem[Zou and Hastie(2005)]{zou2003regression}
Hui Zou and Trevor Hastie.
\newblock {Regularization and Variable Selection via the Elastic Net}.
\newblock \emph{Journal of the Royal Statistical Society: Series B},
  67\penalty0 (2):\penalty0 301--320, 2005.

\bibitem[Zou et~al.(2006)Zou, Hastie, and Tibshirani]{zou2006sparse}
Hui Zou, Trevor Hastie, and Robert Tibshirani.
\newblock {Sparse Principal Component Analysis}.
\newblock \emph{Journal of Computational and Graphical Statistics}, 15\penalty0
  (2):\penalty0 265--286, 2006.

\end{thebibliography}
